\documentclass[11pt]{article}
\usepackage{fullpage}
\usepackage{fixltx2e}
\usepackage{amsmath,amsthm}
\usepackage{amssymb}
\usepackage{times}
\usepackage{paralist}
\usepackage{mathrsfs}
\usepackage[latin1]{inputenc}
\usepackage{framed}
\usepackage[ruled,vlined]{algorithm2e}
\usepackage{multirow}
\usepackage{wrapfig}
\usepackage{tikz}
\usetikzlibrary{arrows,automata,shapes,decorations,calc,matrix,decorations.pathmorphing}
\newcommand{\val}{\mbox{\rm val}}
\newcommand{\cost}{\mbox{\rm R}}

\newcommand{\RR}{\mbox{\rm TB}}
\newcommand{\Red}{\mbox{\rm Red}}
\newcommand{\PromValExtErg}{\mbox{\sc PromValErg}}
\newcommand{\PTIME}{\mbox{\sf PTIME}}
\newcommand{\PSPACE}{\mbox{\sf PSPACE}}
\newcommand{\PPAD}{\mbox{\sf PPAD}}
\newcommand{\FNP}{\mbox{\sf FNP}}
\newcommand{\NP}{\mbox{\sf NP}}
\newcommand{\coNP}{\mbox{\sf coNP}}
\newcommand{\FPTAS}{\mbox{\sf FPTAS}}

\newcommand{\Proc}{\mathsf{Proc}}
\newcommand{\D}{{\mathcal D}}
\newcommand{\ProbDist}{\mathsf{ProbDist}}

\newcommand{\Exp}{\mathbb{E}}

\newcommand{\Z}{\ensuremath{{\rm \mathbb Z}}}
\newcommand{\N}{\ensuremath{{\rm \mathbb N}}}

\newcommand{\var}{\mathsf{var}}

\newcommand{\ExpRew}{\mathsf{ExpRew}}

\newcommand{\act}{A}
\newcommand{\mov}{\Gamma}
\newcommand{\trans}{\delta}
\newcommand{\stra}{\sigma}
\newcommand{\bigstra}{\Sigma}
\newcommand\distr{{\mathcal D}}
\newcommand\pat{\pi}
\newcommand\pats{\Pi}
\newcommand{\game}{G}
\newcommand{\supp}{\mathrm{Supp}}

\newcommand{\cala}{{\mathcal A}}
\newcommand{\Avg}{\mathsf{Avg}}
\newcommand{\LimInfAvg}{\mathsf{LimInfAvg}}
\newcommand{\LimSupAvg}{\mathsf{LimSupAvg}}
\def\set#1{\{ #1 \}}

\newcommand{\ov}{\overline}

\newcommand{\UH}{\underline{H}}
\newcommand{\OH}{\overline{H}}

\makeatletter

\begingroup \catcode `|=0 \catcode `[= 1
\catcode`]=2 \catcode `\{=12 \catcode `\}=12
\catcode`\\=12 |gdef|@xcomment#1\end{comment}[|end[comment]]
|endgroup

\def\@comment{\let\do\@makeother \dospecials\catcode`\^^M=10\def\par{}}

\def\begincomment{\@comment\@xcomment}

\makeatother

\newenvironment{comment}{\begincomment}{}

 \newtheorem{theorem}{Theorem}
 
 \newtheorem{corollary}[theorem]{Corollary}
 \newtheorem{lemma}[theorem]{Lemma}
 \newtheorem{remark}[theorem]{Remark}

\newcommand{\TrueSat}{\mbox{\rm TrueSat}}
\newcommand{\Sat}{\mbox{\rm Sat}}
\newcommand{\VarHoffmanKarp}{\mbox{\sc VarHoffmanKarp}}
\newcommand{\poly}{\mbox{\rm POLY}}
\newcommand{\BestStrategyProg}{\mbox{\rm BestStrategyProg}}

\author{
Krishnendu Chatterjee\thanks{IST Austria. Email: {\tt krish.chat@ist.ac.at}}
\and Rasmus Ibsen-Jensen\thanks{Department of Computer Science, Aarhus University, Denmark. E-mail:
{\tt rij@cs.au.dk}
}}

\begin{document}
\title{The Complexity of Ergodic Mean-payoff Games\thanks{
The first author was supported by 
FWF Grant No P 23499-N23,  FWF NFN Grant No S11407-N23 (RiSE), ERC Start grant (279307: Graph Games), and Microsoft faculty fellows award. Work of the second author supported by the Sino-Danish Center for the Theory of Interactive Computation,
funded by the Danish National Research Foundation and the National
Science Foundation of China (under the grant 61061130540). The second author acknowledge support from the Center for research in
the Foundations of Electronic Markets (CFEM), supported by the Danish
Strategic Research Council.}}
\date{}
\maketitle

\begin{abstract}
We study  two-player (zero-sum) concurrent mean-payoff games 
played on a finite-state graph.
We focus on the important sub-class of ergodic games where all states are 
visited infinitely often with probability~1.
The algorithmic study of ergodic games was initiated in a seminal work 
of Hoffman and Karp in 1966, but all basic complexity questions have remained 
unresolved.
Our main results for ergodic games are as follows: We establish
(1)~an optimal exponential bound on the patience of stationary 
strategies (where patience of a distribution is the inverse of the smallest
positive probability and represents a complexity measure of a stationary 
strategy);
(2)~the approximation problem lies in $\FNP$;
(3)~the approximation problem is at least as hard as the decision problem 
for simple stochastic games (for which $\NP \cap \coNP$ is the long-standing
best known bound).
We present a variant of the strategy-iteration algorithm by Hoffman and Karp; 
show that both our algorithm and the classical value-iteration algorithm can 
approximate the value in exponential time; and 
identify a subclass where the value-iteration algorithm is a $\FPTAS$.
We also show that the exact value can be expressed in the existential
theory of the reals, and establish square-root sum hardness for a related
class of games.
\end{abstract}


\section{Introduction}

\noindent{\bf Concurrent games.}
Concurrent games are played over finite-state graphs by two players 
(Player~1 and Player~2) for an infinite number of rounds.
In every round, both players simultaneously choose moves (or actions), 
and the current state and the joint moves determine a probability 
distribution over the successor states. 
The outcome of the game (or a \emph{play}) is an infinite sequence of states
and action pairs.
Concurrent games were introduced in a seminal work by Shapley~\cite{Sha53}, and 
they are the most well-studied game models in stochastic graph games, with many important special cases.

\smallskip\noindent{\bf Mean-payoff (limit-average) objectives.}
The most fundamental objective for concurrent games is the \emph{limit-average}
(or mean-payoff) objective, where a reward is associated to every transition 
and the payoff of a play is the limit-inferior (or limit-superior) average 
of the rewards of the play. 
The original work of Shapley~\cite{Sha53} considered \emph{discounted} sum 
objectives (or games that stop with probability~1); and the class of 
concurrent games with limit-average objectives (or games that have zero stop 
probabilities) was introduced by Gillette in~\cite{Gil57}.
The Player-1 \emph{value} $\val(s)$ of the game at a state $s$ is the 
supremum value of the expectation that Player~1 can guarantee for the 
limit-average objective against all strategies of Player~2.
The games are zero-sum, so the objective of Player~2 is the opposite.
The study of concurrent mean-payoff games and its sub-classes have 
received huge attention over the last decades, both for mathematical
results as well as algorithmic studies.
Some key celebrated results are as follows:
(1)~the existence of values (or determinacy or equivalence of switching of 
strategy quantifiers for the players as in von-Neumann's min-max theorem) for 
concurrent discounted games was established in~\cite{Sha53};
(2)~the result of Blackwell and Ferguson established existence of values for
the celebrated game of Big-Match~\cite{BF68}; and
(3)~developing on the results of~\cite{BF68} and Bewley and Kohlberg on Puisuex 
series~\cite{BK76} the existence of values for concurrent mean-payoff games was 
established by Mertens and Neyman~\cite{MN81}.

\smallskip\noindent{\bf Sub-classes.} The general class of concurrent mean-payoff
games is notoriously difficult for algorithmic analysis. 
The current best known solution for general concurrent mean-payoff games 
is achieved by a reduction to the theory of the reals over addition and 
multiplication with three quantifier alternations~\cite{CMH08} (also 
see~\cite{HKLMT11} for a better reduction for constant state spaces).
The strategies that are required in general for concurrent mean-payoff games 
are infinite-memory strategies that depend in a complex way on the history 
of the game~\cite{MN81,BF68}, 
and analysis of such strategies make the algorithmic study complicated.
Hence several sub-classes of concurrent mean-payoff games have been studied 
algorithmically both in terms of restrictions of the graph structure and
restrictions of the objective.
The three prominent restrictions in terms of the graph structure are as follows:
(1)~\emph{Ergodic games (aka irreducible games)} where every state is visited infinitely often almost-surely.
(2)~\emph{Turn-based stochastic games}, where in each state at most one player can choose 
between multiple moves.
(3)~\emph{Deterministic games}, where the transition functions are deterministic.
The most well-studied restriction in terms of objective is the \emph{reachability} 
objectives.
A reachability objective consists of a set $U$ of \emph{terminal} states 
(absorbing or sink states that are states with only self-loops),
such that the set $U$ is exactly the set of states where out-going transitions 
are assigned reward~1 and all other transitions are assigned reward~0.
For all these sub-classes, except deterministic mean-payoff games (that is ergodic mean-payoff games, concurrent reachability 
games, and turn-based stochastic mean-payoff games) \emph{stationary} strategies are 
sufficient, where a stationary strategy is independent of the past history of 
the game and depends only on the current state.

\begin{wrapfigure}{r}{0.4\linewidth}
\vspace{-12mm}
\begin{center}
\begin{tikzpicture}[node distance=3cm]
\matrix (v1) [label=above:$t:$,minimum height=1.5em,minimum width=1.5em,matrix of math nodes,nodes in empty cells, left delimiter={.},right delimiter={.}]
{
\\
};
\draw[black] (v1-1-1.north west) -- (v1-1-1.north east);
\draw[black] (v1-1-1.south west) -- (v1-1-1.south east);
\draw[black] (v1-1-1.north west) -- (v1-1-1.south west);
\draw[black] (v1-1-1.north east) -- (v1-1-1.south east);
\matrix (s) [label=above left:$s:$,left of=v1,minimum height=1.5em,minimum width=1.5em,matrix of math nodes,nodes in empty cells, left delimiter={.},right delimiter={.}]
{
&\\
&\\
};
\draw[black] (s-1-1.north west) -- (s-1-2.north east);
\draw[black] (s-1-1.south west) -- (s-1-2.south east);
\draw[black] (s-2-1.south west) -- (s-2-2.south east);
\draw[black] (s-1-1.north west) -- (s-2-1.south west);
\draw[black] (s-1-2.north west) -- (s-2-2.south west);
\draw[black] (s-1-2.north east) -- (s-2-2.south east);

\node[anchor=south east, left=1pt] (s-0-0) at (s-1-1.north west) {};
\foreach[count=\i] \v in {\(A\),\(B\)}{
  \node (s-\i-0) at (s-0-0 |- s-\i-\i) {};
  \path (s-\i-0.north) -- (s-\i-0.south) node [midway, left] { $a_{\i}$ };
}

\node[anchor=north east, below=1pt] (s-3-3) at (s-2-2.south east) {};
\foreach[count=\i] \v in {\(A\),\(B\)}{
  \node (s-\i-3) at (s-\i-\i |- s-3-3) {};
  \path (s-\i-3.east) -- (s-\i-3.west) node [midway, below] { $b_{\i}$ };
}

\draw[->,dashed](s-1-1.center) .. controls +(45:1) and +(90:1).. node[midway,above] (x) {2} (s-1-1.north);
\draw[->,dashed,out=45,in=130](s-1-1.center) 
 to node[midway,above] (x) {2} (v1);
\draw[->,out=65,in=155](s-1-2.center) to node[midway,below] (x) {1}  (v1);
\draw[->,out=-65,in=-155](s-2-2.center) to node[midway,above] (x) {2}  (v1);
\draw[->,out=-45,in=-130](s-2-1.center) to node[midway,below] (x) {1} (v1);
\draw[->](v1.center) to node[above,pos=0.55] (x) {2}(s);
\end{tikzpicture}
\vspace{-6mm}
\caption{Example game $G$.
\label{fig:ergodic sqrt intro}
}
\end{center}
\vspace{-6mm}
\end{wrapfigure}
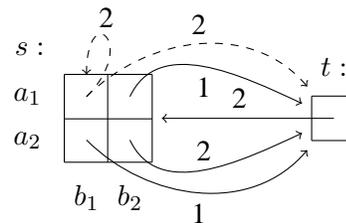

\smallskip\noindent{\bf An example.}
Consider the ergodic mean-payoff game shown in Figure~\ref{fig:ergodic sqrt intro}.
All transitions other than the dashed edges have probability~1, and each 
dashed edge has probability~$1/2$.
The transitions are annotated with the rewards.
The stationary optimal strategy for both players is to play the first action
($a_1$ and $b_1$ for Player~1 and Player~2, respectively)  
with probability $4-2\cdot \sqrt{3}$ in state $s$, and this ensures that the value is $\sqrt{3}$.

\smallskip\noindent{\bf Previous results.}
The decision problem of whether the value of the game at a state is at least a 
given threshold for turn-based stochastic reachability games (and 
also turn-based mean-payoff games with deterministic transition function) 
lie in $\NP \cap \coNP$~\cite{Con92,ZP96}.
They are among the rare and intriguing combinatorial problems that lie 
in $\NP \cap  \coNP$, but not known to be in $\PTIME$.
The existence of polynomial-time algorithms for the above decision 
questions are long-standing open problems.
The algorithmic solution for turn-based games that is most efficient in 
practice is the \emph{strategy-iteration} algorithm, where the 
algorithm iterates over local improvement of strategies which is
then established to converge to a globally optimal 
strategy.
For ergodic games, Hoffman and Karp~\cite{HK} presented a strategy-iteration
algorithm and also established that stationary strategies are sufficient
for such games.
For concurrent reachability games, again stationary strategies are 
sufficient (for $\epsilon$-optimal strategies, for all 
$\epsilon>0$)~\cite{Eve57,dAM01}; the decision problem is in $\PSPACE$ and 
\emph{square-root sum} hard~\cite{EY}.\footnote{The square-root sum problem 
is an important problem from computational geometry, where given a set of 
natural numbers $n_1,n_2,\ldots,n_k$, 
the question is whether the sum of the square roots exceed an integer $b$. 
The square root sum problem is not known to be in $\NP$.}

\smallskip\noindent{\bf Key intriguing complexity questions.}
There are several key intriguing open questions related to the 
complexity of the various sub-classes of concurrent mean-payoff games.
Some of them are as follows:
(1)~Does there exist a sub-class of concurrent mean-payoff games where the 
approximation problem is simpler than the exact decision problem, e.g., 
the decision problem is square-root sum hard, but the approximation problem can 
be solved in $\FNP$?
(2)~There is no convergence result associated with the two classical algorithms,
namely the strategy-iteration algorithm of  Hoffman and Karp, 
and the value-iteration algorithm, for ergodic games; and is it possible to 
establish a convergence for them for approximating the values of ergodic games.
(3)~The complexity of a stationary strategy is described by its \emph{patience}
which is the inverse of the minimum non-zero probability assigned to a move~\cite{Eve57},
and there is no bound known for the patience of stationary strategies for 
ergodic games.

\smallskip\noindent{\bf Our results.} 
The study of the ergodic games was initiated in the seminal work of 
Hoffman and Karp~\cite{HK}, and most of the complexity questions
(related to computational-, strategy-, and algorithmic-complexity) 
have remained open. 
In this work we focus on the complexity of simple generalizations of 
ergodic games (that subsume ergodic games).
Ergodic games form a very important sub-class of concurrent games
subsuming the special cases of uni-chain Markov decision processes
and uni-chain turn-based stochastic games (that have been studied in great 
depth in the literature with numerous applications, see~\cite{FV97,Puterman}).
We consider generalizations of ergodic games called 
\emph{sure} ergodic games where all plays are guaranteed to reach an ergodic 
component (a sub-game that is ergodic); and 
\emph{almost-sure} ergodic games where with probability~1 an ergodic component 
is reached.
Every ergodic game is sure ergodic, and every sure ergodic game is almost-sure
ergodic.
Intuitively the generalizations allow us to consider that after a finite prefix
an ergodic component is reached.

\begin{enumerate}
\item \emph{(Strategy and approximation complexity).}
We show that for almost-sure ergodic games the optimal bound on patience 
required for $\epsilon$-optimal stationary strategies, for $\epsilon>0$, 
is exponential (we establish the upper bound for almost-sure ergodic games, 
and the lower bound for ergodic games).
We then show that the approximation problem for {\em turn-based} stochastic 
ergodic mean-payoff games is at least as hard as solving the decision 
problem for turn-based stochastic reachability games (aka simple stochastic
games); and finally show that the approximation problem 
belongs to $\FNP$ for almost-sure ergodic games.
Observe that our results imply that improving our $\FNP$-bound
for the approximation problem to polynomial time would require solving 
the long-standing open question of whether the decision problem of 
turn-based stochastic reachability games can be solved in polynomial time.

\item \emph{(Algorithm).} We present a variant of the Hoffman-Karp algorithm and 
show that for all $\epsilon$-approximation (for $\epsilon>0$) our algorithm converges 
with in exponential number of iterations for almost-sure ergodic games.
Again our result is optimal, since even for turn-based stochastic reachability
games the strategy-iteration algorithms require exponential iterations~\cite{Fr11,Fe10}.
We analyze the value-iteration algorithm for ergodic games and show that for all 
$\epsilon>0$, the value-iteration algorithm requires at most 
$O(\UH \cdot W \cdot \epsilon^{-1} \cdot \log (\epsilon^{-1}))$ iterations,
where $\UH$ is the upper bound on the expected hitting time of state pairs that 
Player~1 can ensure and $W$ is the maximal reward value.
We show that $\UH$ is at most $n\cdot (\delta_{\min})^{-n}$, where $n$ is the number
of states of the game, and $\delta_{\min}$ the smallest positive transition 
probability.
Thus our result establishes an exponential upper bound for the value-iteration 
algorithm for approximation.
This result is in sharp contrast to concurrent reachability games where the 
value-iteration algorithm requires double exponentially many steps~\cite{HIM11}.
Observe that  we have a polynomial-time approximation scheme if $\UH$ is 
polynomial and the numbers $W$ and $\epsilon$ are represented in unary.
Thus we identify a subclass of ergodic games where the value-iteration 
algorithm is polynomial (see Remark~\ref{remark:val-iter} for further details).

\item \emph{(Exact complexity).} We show that the exact decision problem for 
almost-sure ergodic games can be expressed in the existential theory of the reals (in contrast to 
general concurrent mean-payoff games where quantifier alternations are required).
Finally, we show that the exact decision problem for sure ergodic games is 
square-root sum hard.

\end{enumerate} 

\noindent{\bf Technical contribution and remarks.} 
Our main result is establishing the optimal bound of exponential patience for
$\epsilon$-optimal stationary strategies, for $\epsilon>0$, in almost-sure ergodic games. 
Our result is in sharp contrast to the optimal bound of double-exponential 
patience for concurrent reachability games~\cite{HKM09}, and also the 
double-exponential iterations required by the strategy-iteration and 
the value-iteration algorithms for concurrent reachability games~\cite{HIM11}.
Our upper bound on the exponential patience is achieved by a coupling 
argument.
While coupling argument is a well-established tool in probability theory, 
to the best of our knowledge the argument has not been used for concurrent 
mean-payoff games before.
Our lower bound example constructs a family of ergodic mean-payoff games where 
exponential patience is required. 
Our results provide a complete picture for almost-sure and sure ergodic games 
(subsuming ergodic games) in terms of strategy complexity, computational complexity,
and algorithmic complexity; and present answers to some of the key intriguing 
open questions related to the computational complexity of concurrent mean-payoff games.

\smallskip\noindent{\bf Comparison with results for Shapley games.}
For Shapley (concurrent discounted) games, the exact decision problem is 
square-root sum hard~\cite{EY10}, and the fact that the approximation 
problem is in $\FNP$ is straight-forward to prove\footnote{The basic argument 
is to show that for $\epsilon$-approximation, for $\epsilon>0$, in discounted games,
the players need to play optimally only for exponentially many steps, and hence 
a strategy with exponential patience for $\epsilon$-approximation can be constructed.
For details, see~\cite[Lemma~6, Section~1.10]{I13}: we thank Peter Bro Miltersen
for this argument.}.
The more interesting and challenging question is whether the approximation problem can 
be solved in $\PPAD$. 
The $\PPAD$ complexity for the approximation problem for Shapley games was 
established in~\cite{EY10}; and the $\PPAD$ complexity arguments use the 
existence of unique (Banach) fixpoint (due to contraction mapping) and the 
fact that weak approximation implies strong approximation.
A $\PPAD$ complexity result for the class of ergodic games (in particular, 
whether weak approximation implies strong approximation) is a subject for 
future work.
Another interesting direction of future work would be to extend our results 
for concurrent games where the values of all states are very close together; and 
for this class of games existence of near optimal stationary strategies was established in~\cite{TechRpt}.

\newcommand{\rand}{r}
\section{Definitions}

In this section we present the definitions of game structures, strategies, 
mean-payoff function, values, and other basic notions.

\smallskip\noindent{\bf Probability distributions.}
For a finite set~$A$, a {\em probability distribution\/} on $A$ is a
function $\trans\!:A\to[0,1]$ such that $\sum_{a \in A} \trans(a) = 1$.
We denote the set of probability distributions on $A$ by $\distr(A)$. 
Given a distribution $\trans \in \distr(A)$, we denote by $\supp(\trans) = 
\{x\in A \mid \trans(x) > 0\}$ the {\em support\/} of the distribution 
$\trans$. We denote by $\rand$ the number of \emph{random} states where the 
transition function is not deterministic, i.e., 
$\rand=|\set{s \in S \mid \exists a_1 \in \Gamma_1(s),
a_2 \in \Gamma_2(s). |\supp(\trans(s,a_1,a_2))| \geq 2}|$.

\smallskip\noindent{\bf Concurrent game structures.} 
A {\em concurrent stochastic game structure\/} 
$\game =  (S, \act,\mov_1, \mov_2, \trans)$ has the 
following components.

\begin{compactitem}

\item A finite state space $S$ and a finite set $\act$ of actions (or moves).

\item Two move assignments $\mov_1, \mov_2 \!: S\to 2^{\act}
	\setminus \emptyset$.  For $i \in \{1,2\}$, assignment
	$\mov_i$ associates with each state $s \in S$ the non-empty
	set $\mov_i(s) \subseteq \act$ of moves available to Player~$i$
	at state $s$.  

\item A probabilistic transition function
	$\trans\!:S\times\act\times\act \to \distr(S)$, which
	associates with every state $s \in S$ and moves $a_1 \in
	\mov_1(s)$ and $a_2 \in \mov_2(s)$, a probability
	distribution $\trans(s,a_1,a_2) \in \distr(S)$ for the
	successor state.
\end{compactitem}
We denote by $\trans_{\min}$ the minimum non-zero transition 
probability, i.e., $\trans_{\min}=\min_{s,t \in S} \min_{a_1 \in \mov_1(s),a_2\in \mov_2(s)}
\set{\trans(s,a_1,a_2)(t) \mid \trans(s,a_1,a_2)(t)>0}$.
We denote by $n$ the number of states (i.e., $n=|S|$), and by 
$m$ the maximal number of actions available for a player at a state 
(i.e., $m=\max_{s\in S} \max\set{|\mov_1(s)|,|\mov_2(s)|}$). We denote by $\rand$ the number of \emph{random} states where the 
transition function is not deterministic, i.e., 
$\rand=|\set{s \in S \mid \exists a_1 \in \Gamma_1(s),
a_2 \in \Gamma_2(s). |\supp(\trans(s,a_1,a_2))| \geq 2}|$.

\medskip\noindent{\bf Plays.}
At every state $s\in S$, Player~1 chooses a move $a_1\in\mov_1(s)$,
and simultaneously and independently
Player~2 chooses a move $a_2\in\mov_2(s)$.  
The game then proceeds to the successor state $t$ with probability
$\trans(s,a_1,a_2)(t)$, for all $t \in S$. 
A {\em path\/} or a {\em play\/} of $\game$ is an infinite sequence
$\pat =\big( (s_0,a^0_1, a^0_2), (s_1, a^1_1, a^1_2), (s_2,a_1^2,a_2^2)\ldots\big)$ of states and action pairs such that for all 
$k\ge 0$ we have (i)~$a^k_1 \in \mov_1(s_k)$ and $a^k_2 \in \mov_2(s_k)$; and 
(ii)~$s_{k+1} \in \supp(\trans(s_k,a^k_1,a^k_2))$.
We denote by $\pats$ the set of all paths.

\smallskip\noindent{\bf Strategies.}
A {\em strategy\/} for a player is a recipe that describes how to 
extend prefixes of a play.
Formally, a strategy for Player~$i\in\{1,2\}$ is a mapping 
$\stra_i\!:(S\times \act \times \act)^* \times S \to\distr(\act)$ that associates with every 
finite sequence $x \in (S\times \act \times \act)^*$  of state and action pairs, and the 
current state $s$ in $S$, representing the past history of the game, 
a probability distribution $\stra_i(x \cdot s)$ used to select
the next move. 
The strategy $\stra_i$ can prescribe only moves that are available to Player~$i$;
that is, for all sequences $x\in (S \times \act \times \act)^*$ and states $s\in S$, we require that
$\supp(\stra_i(x\cdot s)) \subseteq \mov_i(s)$.  
We denote by $\bigstra_i$ the set of all strategies for Player $i\in\{1,2\}$.
Once the starting state $s$ and the strategies $\stra_1$ and $\stra_2$
for the two players have been chosen, 
then we have a random walk $\pat_s^{\stra_1,\stra_2}$ for which 
the probabilities of events are uniquely defined~\cite{VardiP85}, where an {\em
event\/} $\cala\subseteq\pats$ is a measurable set of
paths.
For an event $\cala\subseteq\pats$, we denote by $\Pr_s^{\stra_1,\stra_2}(\cala)$ 
the probability that a path belongs to $\cala$ when the game starts from $s$ and 
the players use the strategies $\stra_1$ and~$\stra_2$;
and denote $\Exp_{s}^{\stra_1,\stra_2}[\cdot]$ as the associated 
expectation measure.
We consider in particular stationary and positional strategies. 
%
A strategy $\sigma_i$ is \emph{stationary} (or memoryless) if it is independent of the history 
but only depends on the current state, i.e., for all $x,x'\in(S\times A\times A)^*$ and all $s\in S$, 
we have $\sigma_i(x\cdot s)=\sigma_i(x'\cdot s)$, and thus can be expressed as a function 
$\stra_i: S \to \distr(\act)$. 
For stationary strategies, the complexity of the strategy is described by the 
\emph{patience} of the strategy, which is the inverse of the minimum non-zero 
probability assigned to an action~\cite{Eve57}. 
Formally, for a stationary strategy  $\stra_i:S \to \distr(\act)$ for Player~$i$, 
the patience is  
$\max_{s \in S} \max_{a \in \mov_i(s)} \set{ \frac{1}{\stra_i(s)(a)} \mid \stra_i(s)(a)>0}$.
A strategy is \emph{pure (deterministic)} if it does not use randomization, i.e., for any history there is always
some unique action $a$ that is played with probability~1.
A pure stationary strategy $\stra_i$ is also called a {\em positional} strategy, and 
represented as a function $\stra_i: S \to \act$.
We call a pair of strategies $(\stra_1,\stra_2) \in \bigstra_1 \times \bigstra_2$ 
a \emph{strategy profile}.

\smallskip\noindent{\bf The mean-payoff function.}
In this work we consider maximizing \emph{limit-average} (or mean-payoff) functions
for Player~1, and the objective of Player~2 is opposite (i.e., the games are zero-sum). 
We consider concurrent games with a reward function $\cost: S \times \act \times 
\act \to [0,1]$ that assigns a reward value $0\leq \cost(s,a_1,a_2)\leq 1$ 
for all $s\in S$, $a_1 \in \mov_1(s)$, and $a_2 \in \mov_2(s)$.
For a path $\pat= \big((s_0, a^0_1, a^0_2), (s_1, a^1_1,a^1_2), \ldots\big)$, 
the average for $T$ steps is $\Avg_T(\pat)= \frac{1}{T} \cdot \sum_{i=0}^{T-1} \cost(s_i,a^i_1,a^i_2)$,
and the limit-inferior average (resp. limit-superior average) is defined as 
follows:
$\LimInfAvg(\pat)= \lim\inf_{T \to \infty} \Avg_T$ 
(resp. $\LimSupAvg(\pat)= \lim\sup_{T \to \infty} \Avg_T$).
For brevity we denote concurrent games with mean-payoff functions as CMPGs (concurrent 
mean-payoff games).

\smallskip\noindent{\bf Values and $\epsilon$-optimal strategies.}
Given a CMPG $G$ and a reward function $\cost$, the \emph{lower value} 
$\underline{v}_s$ (resp. the \emph{upper value} $\ov{v}_s$) at a state $s$ 
is defined as follows:
\[
\underline{v}_s = \sup_{\stra_1 \in \bigstra_1 } \inf_{\stra_2 \in \bigstra_2} 
\Exp_s^{\stra_1,\stra_2}[\LimInfAvg];
\qquad 
\ov{v}_s = \inf_{\stra_2 \in \bigstra_2} \sup_{\stra_1 \in \bigstra_1} 
\Exp_s^{\stra_1,\stra_2}[\LimSupAvg].
\]
The celebrated result of Mertens and Neyman~\cite{MN81} shows that the upper and lower
value coincide and gives the \emph{value} of the game denoted as $v_s$.
For $\epsilon\geq 0$, a strategy $\stra_1$ for Player~1 is \emph{$\epsilon$-optimal} 
if we have $v_s - \epsilon \leq \inf_{\stra_2\in \bigstra_2} \Exp_s^{\stra_1,\stra_2}[\LimInfAvg]$.
An \emph{optimal} strategy is a $0$-optimal strategy.

\smallskip\noindent{\bf Game classes.} We consider the following special classes of 
CMPGs.
\begin{enumerate}
\item {\em Variants of ergodic CMPGs.} Given a CMPG $G$, a set $C$ of states in $G$ is called an \emph{ergodic 
component}, if for all states $s,t \in C$, for all strategy profiles 
$(\stra_1,\stra_2)$, if we start at $s$, then $t$ is visited infinitely 
often with probability~1 in the random walk $\pat_s^{\stra_1,\stra_2}$.
A CMPG is \emph{ergodic} if the set $S$ of states is an ergodic component.
A CMPG is \emph{sure ergodic} if for all strategy profiles $(\stra_1,\stra_2)$
and for all start states $s$, ergodic components are reached certainly (all plays
reach some ergodic component).
A CMPG is \emph{almost-sure ergodic} if for all strategy profiles $(\stra_1,\stra_2)$
and for all start states $s$, ergodic components are reached with probability~1.
Observe that every ergodic CMPG is also a sure ergodic CMPG, and every 
sure ergodic CMPG  is also an almost-sure ergodic CMPG.

\item {\em Turn-based stochastic games, MDPs and SSGs.} 
A game structure $\game$ is {\em turn-based stochastic\/} if at every
state at most one player can choose among multiple moves; that is, for
every state $s \in S$ there exists at most one $i \in \{1,2\}$ with
$|\mov_i(s)| > 1$. 
A game structure is a Player-2 \emph{Markov decision process (MDP)} if for all 
$s \in S$ we have $|\mov_1(s)|=1$, i.e., only Player~2 has choice of 
actions in the game, and Player-1 MDPs are defined analogously.
A \emph{simple stochastic game (SSG)}~\cite{Con92} is an almost-sure ergodic turn-based 
stochastic game with two ergodic components, where both the ergodic
components (called terminal states) are a single \emph{absorbing state} 
(an absorbing state has only a self-loop transition); 
one terminal state ($\top$) has reward~1 and the other terminal 
state ($\bot$) has reward~0; and all positive transition probabilities 
are either~$\frac{1}{2}$ or~1.
The almost-sure reachability property to the ergodic components
for SSGs is referred to as the \emph{stopping} property~\cite{Con92}.
\end{enumerate}

\begin{remark}
The results of Hoffman and Karp~\cite{HK} established that for ergodic 
CMPGs \emph{optimal stationary} strategies exist (for both players).
Moreover, for an ergodic CMPG the value for every state is the same,
which is called the value of the game.
We argue that the result for existence of optimal stationary strategies also 
extends to almost-sure ergodic CMPGs.
Consider an almost-sure ergodic CMPG $G$. Notice first that in the 
ergodic components, there exist optimal stationary strategies, as 
shown by Hoffman and Karp~\cite{HK}. Notice also that eventually 
some ergodic component is reached with probability~1 after a finite number 
of steps, and therefore that we can ignore the rewards of the finite prefix 
(since mean-payoff functions are independent of finite prefixes). 
Hence, we get an almost-sure reachability game, in the states which are not 
in the ergodic components, by considering any ergodic component $C$ to be a 
terminal with reward equal to the value of $C$. 
In such games it is easy to see that there exist optimal stationary strategies. 
\end{remark}

\noindent{\bf Value and the approximation problem.} 
Given a CMPG 
$G$, a state $s$ of $G$, and a rational
threshold $\lambda$, the \emph{value} problem is the decision problem that 
asks whether $v_s$ is at most $\lambda$.
Given a CMPG 
$G$, a state $s$ of $G$,  and a tolerance 
$\epsilon>0$,  the \emph{approximation} problem asks to compute an interval
of length $\epsilon$ such that the value $v_s$ lies in the interval.
We present the formal definition of the decision version of the 
approximation problem in Section~\ref{subsec:approx}.
In the following sections we consider the value problem and the 
approximation problem for almost-sure ergodic, sure ergodic, and 
ergodic games.

\section{Complexity of Approximation for Almost-sure Ergodic Games}\label{sec:approx}
In this section we present three results for almost-sure ergodic games:
(1)~First we establish (in Section~\ref{subsec:strategy}) an optimal 
exponential bound on the patience of $\epsilon$-optimal stationary strategies, 
for all $\epsilon>0$. 
(2)~Second we show (in Section~\ref{subsec:ssg hard}) that the approximation 
problem (even for turn-based stochastic ergodic mean-payoff games) is at least 
as hard as solving the value problem for SSGs.
(3)~Finally, we show (in Section~\ref{subsec:approx}) that the approximation 
problem lies in $\FNP$.



\subsection{Strategy complexity}\label{subsec:strategy}
In this section we present results related to $\epsilon$-optimal 
stationary strategies for almost-sure ergodic CMPGs, that on one hand 
establishes an optimal exponential bound for patience, and on the other hand 
is used to establish the complexity of approximation of values in the 
following subsection.
The results of this section is also used in the algorithmic analysis
in Section~\ref{sec:strategy iteration ergodic}.
We start with the notion of $q$-rounded strategies.

\smallskip\noindent{\bf The classes of $q$-rounded distributions and strategies.}
For $q \in \N$, a distribution $d$ over a finite set $Z$ is a {\em $q$-rounded distribution} if for all $z\in Z$ we have that $d(z)=\frac{p}{q}$ for some number $p\in \N$. A stationary strategy $\sigma$ is a {\em $q$-rounded strategy}, if for all states $s$ the distribution $\sigma(s)$ is a $q$-rounded distribution.

\smallskip\noindent{\bf Patience.} 
Observe that the patience of a $q$-rounded strategy is at most $q$.
We show that for almost-sure ergodic CMPGs for all $\epsilon>0$ there are 
$q$-rounded $\epsilon$-optimal strategies, where $q$ is as follows:   
\[
\left\lceil  4\cdot \epsilon^{-1}\cdot m \cdot n^2\cdot (\delta_{\min})^{-\rand}\right\rceil\enspace .
\]
This immediately implies an exponential upper bound on the patience.
We start with a lemma related to the probability of reaching states that 
are guaranteed to be reached with positive probability.


\begin{lemma}\label{lemm:prob to reach}
Given a CMPG $G$, let $s$ be a state in $G$, and $T$ be a set of states
such that for all strategy profiles the set $T$ is reachable (with positive 
probability) from $s$.
For all strategy profiles the probability to reach $T$ from $s$ in $n$ steps 
is at least $(\delta_{\min})^{\rand}$ (where $\rand$ is the number of random 
states).
\end{lemma}
\begin{proof}
The basic idea of the proof is to consider a turn-based deterministic game where 
one player is Player~1 and Player~2 combined, and the opponent makes the choice for the 
probabilistic transitions.
(The formal description of the turn-based deterministic game is as follows: 
$(S \cup (S\times A_1 \times A_2), (A_1 \times A_2) \cup S \cup \set{\bot}, \ov{\Gamma}_1,\ov{\Gamma}_2,\ov{\trans})$;
where for all $s \in S$ and $a_1 \in \Gamma_1(s)$ and $a_2 \in \Gamma_2(s)$ we have 
$\ov{\Gamma}_1(s)=\set{(a_1,a_2) \mid a_1 \in \Gamma_1(s), a_2 \in \Gamma_2(s)}$ and $\ov{\Gamma}_1((s,a_1,a_2))=\set{\bot}$;
$\ov{\Gamma}_2((s,a_1,a_2))=\supp(\trans(s,a_1,a_2))$ and $\ov{\Gamma}_2(s)=\set{\bot}$.
The transition function is as follows: for all $s \in S$ and $a_1 \in \Gamma_1(s)$ and $a_2 \in \Gamma_2(s)$ we have 
$\ov{\trans}(s,(a_1,a_2),\bot)((s,a_1,a_2))=1$ and $\ov{\trans}((s,a_1,a_2),\bot,t)(t)=1$.)
In the turn-based deterministic game, against any strategy of the combined 
players, there is a positional strategy of the player making the probabilistic choices
such that $T$ is reached after being in each state at most once certainly (by positional determinacy for 
turn-based deterministic reachability games~\cite{Thomas97}), as otherwise there would exist a positional 
strategy profile such that $T$ is never reached. 
The probability that exactly the choices made by the positional strategy of the probabilistic player 
in the turn-based deterministic game is executed once in each state in the original game is 
at least $(\delta_{\min})^{\rand}$.
Hence the desired result follows.
\end{proof}

\smallskip\noindent{\bf Variation distance.} 
We use a coupling argument in our proofs and this requires the definition 
of variation distance of two probability distributions.
Given a finite set $Z$, and two distributions $d_1$ and $d_2$ over $Z$, 
the {\em variation distance} of the distributions is 
\[
\var(d_1,d_2)=\frac{1}{2}\cdot \sum_{z\in Z} |d_1(z)-d_2(z)| \enspace .
\]

\smallskip\noindent{\bf Coupling and coupling lemma.}
Let $Z$ be a finite set. For distributions $d_1$ and $d_2$ over the finite set $Z$, 
a {\em coupling} $\omega$ is a distribution over $Z\times Z$, such that for all $z\in Z$ we have 
$\sum_{z'\in Z} \omega(z,z')=d_1(z)$ and also for all $z'\in Z$ we have 
$\sum_{z\in Z} \omega(z,z')=d_2(z')$. 
We only use the second part of coupling lemma~\cite{aldous} which is stated as follows:
\begin{itemize}
\item {\bf (Coupling lemma).}
For a pair of distributions 
$d_1$ and $d_2$, there exists a coupling $\omega$ of $d_1$ and $d_2$, such that for a random variable 
$(X,Y)$ from the distribution $\omega$, we have that $\var(d_1,d_2)=\Pr[X\neq Y]$.
\end{itemize}

We now show that in almost-sure ergodic CMPGs strategies that play actions with probabilities ``close'' 
to what is played by an optimal strategy also achieve values that are ``close'' to the 
values achieved by the optimal strategy. 

\begin{lemma}\label{lemm:approx optimal strategy}
Consider an almost-sure ergodic CMPG and let $\epsilon> 0$ be a real number.
Let $\sigma_1$ be an optimal stationary strategy for Player~1. 
Let $\sigma'_1$ be a stationary strategy for Player~1 s.t. $\sigma'_1(s)(a)\in [\sigma_1(s)(a)-\frac{1}{q};\sigma_1(s)(a)+\frac{1}{q}]$, 
where $q=4\cdot \epsilon^{-1}\cdot m \cdot n^2\cdot (\delta_{\min})^{-\rand}$, for all states $s$ and actions $a\in \Gamma_1(s)$. 
Then the strategy $\sigma'_1$ is an $\epsilon$-optimal strategy. 
\end{lemma}

\begin{proof}
First observe that we can consider $\epsilon\leq 1$, because as the rewards are in the 
interval $[0,1]$ any strategy is an $\epsilon$-optimal strategy for $\epsilon\geq 1$.
The proof is split up in two parts, and the second part uses the first. 
The first part is related to plays starting in an ergodic component;
and the second part is the other case.
In both cases we show that $\sigma'_1$ guarantees a mean-payoff within $\epsilon$ of the mean-payoff guaranteed by $\sigma_1$, 
thus implying the statement. 
Let $\sigma_2$ be a positional best response strategy against $\sigma_1'$.
Our proof is based on a novel \emph{coupling} argument. 
The precise nature of the coupling argument is different in the two parts, 
but both use the following:
For any state $s$, it is clear that the variation distance between $\sigma'_1(s)$ and $\sigma_1(s)$ 
is at most $\frac{|\Gamma_1(s)|}{2\cdot q}$, by definition of $\sigma'_1(s)$. 
For a state $s$, let $d_1^s$ be the distribution over states defined as follows: 
for $t \in S$ we have $d_1^s(t)=\sum_{a_1 \in \Gamma_1(s)} \sum_{a_2 \in \Gamma_2(s)} \trans(s,a_1,a_2)(t) \cdot \sigma_1(s)(a_1) \cdot \sigma_2(s)(a_2)$. 
Define $d_2^s$ similarly using $\sigma_1'(s)$ instead of $\sigma_1(s)$. 
Then $d_1^s$ and $d_2^s$ also have a variation distance of at most $\frac{ |\Gamma_1(s)|}{2\cdot q} \leq \frac{m}{2\cdot q}$.
Let $s_0$ be the start state, and $P=\pat_{s_0}^{\sigma_1,\sigma_2}$ be the random walk from $s_0$, 
where Player~1 follows $\sigma_1$ and Player~2 follows $\sigma_2$.
Also let $P'=\pat_{s_0}^{\sigma_1',\sigma_2}$ be the similar defined walk, except that Player~1 
follows $\sigma'_1$ instead of $\sigma_1$. 
Let $X^i$ be the random variable indicating 
the $i$-th state of $P$, and let $Y^i$ be the similar defined random variable in $P'$ instead of $P$. 

\smallskip\noindent{\bf The state $s_0$ is in an ergodic component.} 
Consider first the case where $s_0$ is part of an ergodic component. 
Irrespective of the strategy profile, all states of the ergodic component 
are visited infinitely often almost-surely (by definition of an ergodic component). Hence, we can apply 
Lemma~\ref{lemm:prob to reach} and obtain that we require at most
 $n\cdot (\delta_{\min})^{\rand}=\frac{\epsilon\cdot q}{4\cdot n\cdot m}$ 
steps in expectation to get from one state of the component to any other state of the component.

\noindent{\em Coupling argument.}
We now construct a coupling argument. We define the coupling using induction. 
First observe that $X^0=Y^0=s_0$ (the starting state).
For $i,j \in \N$, let $a_{i,j}\geq 0$ be the smallest number such that 
$X^{i+1}=Y^{j+1+a_{i,j}}$. 
By the preceding we know that $a_{i,j}$ exists for all $i,j$ with probability~1 and $a_{i,j}\leq \frac{\epsilon\cdot q}{4\cdot n\cdot m}$ in expectation. 
The coupling is done as follows: (1)~(Base case): Couple $X^0$ and $Y^0$. We have that $X^0=Y^0$;
(2)~(Inductive case): (i)~if $X^i$ is coupled to $Y^j$ and $X^i=Y^j=s_i$, then also couple $X^{i+1}$ and $Y^{j+1}$
such that $\Pr[X^{i+1}\neq Y^{j+1}] = \var(d_1^{s_i},d_2^{s_i})$ (using coupling lemma);
(ii)~if $X^i$ is coupled to $Y^j$, but $X^i\neq Y^j$, then 
$X^{i+1}=Y^{j+1+a_{i,j}}=s_{i+1}$ and $X^{i+1}$ is coupled to $Y^{j+1+a_{i,j}}$, and 
we couple $X^{i+2}$ and $Y^{j+2+a_{i,j}}$ 
such that $\Pr[X^{i+2}\neq Y^{j+2+a_{i,j}}] = \var(d_1^{s_{i+1}},d_2^{s_{i+1}})$ (using coupling lemma).
Notice that all $X^i$ are coupled to some $Y^j$ almost-surely; and moreover in expectation $\frac{j}{i}$
is bounded as follows:
\[
\frac{j}{i}\leq 1+\frac{m}{2\cdot q}\cdot \frac{\epsilon\cdot q}{4\cdot n\cdot m}=  1+\frac{\epsilon}{8\cdot n}.
\] 
The expression can be understood as follows: consider $X^i$ being coupled to $Y^j$. 
With probability at most $\frac{m}{2\cdot q}$ they differ. 
In that case $X^{i+1}$ is coupled to $Y^{j+1+a_{i,j}}$. 
Otherwise $X^{i+1}$ is coupled to $Y^{j+1}$. By using our bound on $a_{i,j}$ we get the desired expression. 
For a state $s$, let $f_s$ (resp. $f_s'$) denote the limit-average frequency of $s$ 
given $\sigma_1$ (resp. $\sigma_1'$) and $\sigma_2$.
Then it follows easily that for every state $s$,  we have $|f_s-f_s'| \leq \frac{\epsilon}{8\cdot n}$.
The formal argument is as follows: for every state $s$, consider the reward function $\cost_s$ that assigns reward~1 to
all transitions from $s$ and~0 otherwise; and then it is clear that the difference of the mean-payoffs of $P$ and $P'$ 
is maximized if the mean-payoff of $P$ is $1$ under $\cost_s$ and the rewards of the steps of $P'$ that are not coupled to $P$ are $0$. 
In that case the mean-payoff of $P'$ under $\cost_s$ is at least $\frac{1}{1+\frac{\epsilon}{8\cdot n}}>1-\frac{\epsilon}{8\cdot n}$ (since $1>1-\left(\frac{\epsilon}{8\cdot n}\right)^2=(1+\frac{\epsilon}{8\cdot n})(1-\frac{\epsilon}{8\cdot n})$) in expectation and
 thus the difference between the mean-payoff of $P$ and the mean-payoff of $P'$ under $\cost_s$ is at most $\frac{\epsilon}{8\cdot n}$ in expectation. 
The mean-payoff value if Player~1 follows a stationary strategy $\sigma_1^1$ and Player~2 follows a stationary strategy $\sigma_2^1$, 
such that the frequencies of the states encountered is $f_s^1$, is  $\sum_{s \in S} \sum_{a_1 \in \mov_1(s)} \sum_{a_2 \in \mov_2(s)} 
f_s^1 \cdot \sigma_1^1(s)(a_1) \cdot \sigma_2^1(s)(a_2) \cdot \cost(s,a_1,a_2)$. 
Thus the differences in mean-payoff value when Player~1 follows $\sigma_1$ (resp. $\sigma_1'$) and Player~2 follows the positional strategy 
$\sigma_2$, which plays action $a_2^s$ in state $s$, 
is 
\[
\sum_{s \in S} \sum_{a_1 \in \mov_1(s)} \big(f_s \cdot \sigma_1(s)(a_1)-f_s' \cdot \sigma_1'(s)(a_1)\big)  \cdot \cost(s,a_1,a_2^s)
\]
Since $|f_s-f_s'|\leq \frac{\epsilon}{8\cdot n}$  (by the preceding argument) 
and $|\sigma_1(s)(a_1)-\sigma_1'(s)(a_1)|\leq \frac{1}{q}$ for all $s\in S$ and $a_1\in \mov_1(s)$ (by definition), 
we have the following inequality
\begin{align*}
& \sum_{s \in S} \sum_{a_1 \in \mov_1(s)} 
\big(f_s \cdot \sigma_1(s)(a_1)-f_s' \cdot \sigma_1'(s)(a_1)\big) \cdot   \cost(s,a_1,a_2^s) \\[2ex]
\leq & \quad
\sum_{s \in S} \sum_{a_1 \in \mov_1(s)} 
|f_s \cdot \sigma_1(s)(a_1)-\big(f_s- \frac{\epsilon}{8\cdot n}\big) \cdot \big(\sigma_1(s)(a_1)-\frac{1}{q}\big)|\\[2ex]
= & \quad \sum_{s \in S} \sum_{a _1\in \mov_1(s)} 
|\frac{\epsilon}{8\cdot n}\cdot \sigma_1(s)(a_1)+f_s\cdot\frac{1}{q}-\frac{\epsilon}{8\cdot n\cdot q}| \\[2ex]
\leq & \quad \sum_{s \in S} 
(\frac{\epsilon}{8\cdot n}+\frac{f_s \cdot m}{q} + \frac{\epsilon \cdot m}{8 \cdot n \cdot q}) \\[2ex]
= & \quad \frac{\epsilon}{8} + \frac{m}{q} + \frac{\epsilon \cdot m}{8 \cdot q} \leq \frac{\epsilon}{8} + \frac{\epsilon}{4} + \frac{\epsilon}{8} \quad = \quad \frac{\epsilon}{2}
\end{align*}
The first inequality uses that $\cost(s,a_1,a_2^s)\leq 1$ and the preceding comments on the differences. 
The second inequality uses that (a)~when we sum over $\sigma_1(s)(a_1)$ for all $a_1$, for a fixed $s\in S$, we get $1$;
(b)~$|\Gamma_1(s)| \leq m$.
The following equality uses that $\sum_{s\in S} f_s=1$ since they represent frequencies.
Finally since $4 \cdot m\cdot n\cdot \epsilon^{-1}\leq q$, $\epsilon \leq 1$, and $n \geq 1$ we have
$\frac{m}{q} \leq \frac{\epsilon}{4}$ and $\frac{\epsilon \cdot m}{8 \cdot q} \leq \frac{\epsilon}{32} \leq \frac{\epsilon}{8}$.
The desired inequality is established.


\medskip\noindent{\bf The state $s_0$ is not in an ergodic component.} 
Now consider the case where the start state $s_0$ is not part of  an ergodic component. 
We divide the walks $P$ and $P'$ into two parts. The part inside some ergodic component 
and the part outside all ergodic components. If $P$ and $P'$ ends up in the same ergodic component, 
then the mean-payoff differs by at most $\frac{\epsilon}{2}$ in expectation, by the first part. 
For any pair of strategies the random walk defined from them almost-surely reaches some ergodic component (since we 
consider almost-sure ergodic CMPGs).   
Hence, we can apply Lemma~\ref{lemm:prob to reach} and see that we require at most 
$n\cdot (\delta_{\min})^{\rand}=\frac{\epsilon\cdot q}{4\cdot n\cdot m}$ steps in expectation 
before we reach an ergodic component. 

\noindent{\em Coupling argument.}
To find the probability that they end up in the same component we again make a coupling argument. 
Notice that $X^0=Y^0=s_0$. We now make the coupling using induction. 
(1)~(Base case): Make a coupling between $X^1$ and $Y^1$, such that $\Pr[X^1\neq Y^1]=\var(d^{s_0}_1,d^{s_0}_2)\leq \frac{|\Gamma_1(s_0)|}{2\cdot q}
\leq  \frac{m}{2\cdot q}$ (such a coupling exists by the coupling lemma).
(2)~(Inductive case): Also, if there is a coupling between $X^i$ and $Y^i$ and $X^i=Y^i=s_i$, then also make a coupling between  $X^{i+1}$ and $Y^{i+1}$, such that $\Pr[X^{i+1}\neq Y^{i+1}]=\var(d^{s_i}_1,d^{s_i}_2)\leq \frac{|\Gamma_1(s_i)|}{2\cdot q} \leq  \frac{m}{2\cdot q}$ (such a coupling exists by the coupling lemma).
Let $\ell$ be the smallest number such that $X^{\ell}$ is some state in an ergodic component. In expectation, $\ell$ is at most $\frac{\epsilon\cdot q}{4\cdot n\cdot m}$. The probability that $X^i\neq Y^i$ for some $0\leq i \leq \ell$ is by union bound  at most  $\frac{m }{2\cdot q}\cdot \frac{\epsilon\cdot q}{4\cdot n\cdot m}\leq \frac{\epsilon}{8\cdot n}\leq \frac{\epsilon}{2}$ in expectation. If that is not the case, then $P$ and $P'$ do end up in the same ergodic component. 
In the worst case, the component the walk  $P$ ends up in has value $1$ and the component that the walk $P'$ ends up in (if they differ) has value $0$. 
Therefore, with probability at most $\frac{\epsilon}{2}$ the walk $P'$ ends up in an ergodic component of value 0 (and hence has mean-payoff 0); and otherwise it ends up in the same component as $P$ does 
and thus gets the same mean-payoff as $P$, except for at most $\frac{\epsilon}{2}$, as we established in the first part. 
Thus $P'$ must ensure the same mean-payoff as $P$ except for $\frac{2\epsilon}{2}=\epsilon$. We therefore get that $\sigma'_1$ is an $\epsilon$-optimal strategy (since $\sigma_1$ is optimal).
\end{proof}

We show that for every integer $q'\geq \ell$, for every distribution over $\ell$ elements, there exists a $q'$-rounded distribution ``close'' to it. 
Together with Lemma~\ref{lemm:approx optimal strategy} it shows the existence of $q'$-rounded $\epsilon$-optimal strategies, 
for every integer $q'$ greater than the $q$ defined in Lemma~\ref{lemm:approx optimal strategy}.

\begin{lemma}\label{lemm:round distribution}
Let $d_1$ be a distribution over a finite set $Z$ of size $\ell$. 
Then for all integers $q\geq \ell$ there exists a $q$-rounded distribution $d_2$ over $Z$, 
such that $|d_1(z)-d_2(z)|<\frac{1}{q}$.
\end{lemma}
\begin{proof}
WLOG we consider that $\ell\geq 2$ (since the unique distribution over a singleton set clearly have the desired properties for all integers $q\geq 1$).
Given distribution $d_1$ we construct a witness distribution $d_2$.
There are two cases. Either (i)~there is an element $z\in Z$ such that $\frac{1}{q}\leq d_1(z)\leq 1-\frac{1}{q}$, or (ii)~no such element exists. 
\begin{itemize}
\item We first consider case~(ii), i.e., there exists no element $z$ such that $\frac{1}{q}\leq d_1(z)\leq 1-\frac{1}{q}$.
Consider an element $z^*\in Z$ such that $1-\frac{1}{q}<d_1(z^*)$. 
Precisely only one such element exists in this case since not all $\ell$ elements can have probability strictly less than $\frac{1}{q}\leq \frac{1}{\ell}$,
and no more than one element can have probability strictly  more than $1-\frac{1}{q}\geq \frac{1}{2}$. 
Then let $d_2(z^*)=1$ and $d_2(z)=0$ for all other elements in $Z$. 
This clearly ensures that $|d_1(z)-d_2(z)|<\frac{1}{q}$ for all $z\in Z$ and that $d_2$ is a $q$-rounded distribution.

\item
Now we consider case (i). Let $z^{\ell}$ be an arbitrary element in $Z$ such that $\frac{1}{q}\leq d_1(z^{\ell})\leq 1-\frac{1}{q}$. Let $\{z^1,\dots,z^{\ell-1}\}$ be an arbitrary ordering of the remaining elements. 
We now construct $d_2$ iteratively such that in step $k$ we have assigned probability to $\{z^1,\dots, z^k\}$. 
We establish the following \emph{iterative property}: 
in step $k$ we have that  $\sum_{c=1}^k (d_1(z^c)-d_2(z^c))\in (-\frac{1}{q};\frac{1}{q})$. 
The iteration stops when $k=\ell-1$, and then we assign $d_2(z^{\ell})$ the probability $1-\sum_{c=1}^{\ell-1}d_2(z^c)$. 
For all $1\leq k\leq \ell-1$, the iterative definition of $d_2(z^k)$ is as follows:
\[
d_2(z^k)=
\begin{cases}  
\displaystyle 
\frac{\lfloor q\cdot d_1(z^k)\rfloor}{q} & \text{if $\sum_{c=1}^{k-1} (d_1(z^c)-d_2(z^c))<0$}\\[2ex]
\displaystyle 
\frac{\lceil q\cdot d_1(z^k)\rceil}{q} & \text{if $\sum_{c=1}^{k-1} (d_1(z^c)-d_2(z^c))\geq 0$}
\end{cases}\]
We use the standard convention that the empty sum is~0.
For $1\leq k\leq \ell-1$, observe that (a)~$|d_1(z^k)-d_2(z^k)|< \frac{1}{q}$; and 
(b)~since $d_1(z^k) \in [0,1]$ also $d_2(z^k)$ is in $[0;1]$. 
Moreover, there exists an integer $p$ such that $d_2(z^k)=\frac{p}{q}$.
We have that 
\[
d_1(z^k)-\frac{1}{q}<\frac{\lfloor q\cdot d_1(z^k)\rfloor}{q}\leq d_1(z^k)\leq \frac{\lceil q\cdot d_1(z^k)\rceil}{q} <d_1(z^k)+\frac{1}{q} \qquad (\ddag).
\] 
Thus, if the sum $\sum_{c=1}^{k-1} (d_1(z^c)-d_2(z^c))$ is negative, then we have that 
\[\frac{-1}{q}< \sum_{c=1}^{k-1} (d_1(z^c)-d_2(z^c))\leq \sum_{c=1}^{k} (d_1(z^c)-d_2(z^c))<\sum_{c=1}^{k-1} (d_1(z^c)-d_2(z^c))+\frac{1}{q}<\frac{1}{q}\enspace ,\] 
where the first inequality is the iterative property (by induction for $k-1$); 
the second inequality follows because in this case we have $d_2(z^k)=\displaystyle\frac{\lfloor q\cdot d_1(z^k)\rfloor}{q}\leq d_1(z^k)$ by ($\ddag$);
the third inequality follows since $d_1(z^k)-d_2(z^k) =d_1(z^k) - \displaystyle\frac{\lfloor q\cdot d_1(z^k)\rfloor}{q} < \frac{1}{q}$ by ($\ddag$); 
the final inequality follows since $\sum_{c=1}^{k-1} (d_1(z^c)-d_2(z^c))$ is negative. 
Symmetrically, if the sum $\sum_{c=1}^{k-1} (d_1(z^c)-d_2(z^c))$ is not negative, then we have that 
\[\frac{1}{q}> \sum_{c=1}^{k-1} (d_1(z^c)-d_2(z^c))\geq \sum_{c=1}^{k} (d_1(z^c)-d_2(z^c))>\sum_{c=1}^{k-1} (d_1(z^c)-d_2(z^c))-\frac{1}{q}\geq \frac{-1}{q}\enspace ,\]
using the iterative property (by induction) and the inequalities of ($\ddag$) as in the previous case.
Thus, in either case, we have that $\frac{-1}{q}<\sum_{c=1}^{k} (d_1(z^c)-d_2(z^c))<\frac{1}{q}$, establishing the iterative property
by induction.

Finally we need to consider $z^{\ell}$. 
First, we show that $|d_1(z^{\ell})-d_2(z^{\ell})|< \frac{1}{q}$. 
We have that \[d_2(z^{\ell})=1-\sum_{c=1}^{\ell-1}d_2(z^c)=\sum_{c=1}^{\ell}(d_1(z^c))-\sum_{c=1}^{\ell-1}(d_2(z^c))=d_1(z^{\ell})+\sum_{c=1}^{\ell-1} (d_1(z^c)-d_2(z^c))\enspace .\]
Hence $|d_1(z^{\ell})-d_2(z^{\ell})|< \frac{1}{q}$, by our iterative property. This also ensures that $d_2(z^{\ell})\in [0;1]$, since $d_1(z^{\ell})\in [\frac{1}{q};1-\frac{1}{q}]$, by definition. Thus, $d_2$ is a distribution over $Z$ (since it is clear that $\sum_{z\in Z} d_2(z)=1$, because of the definition of $d_2(z^{\ell})$ and we have shown for all $z\in Z$ that $d_2(z)\in [0;1]$). Since we have ensured that for each $z\in (Z\setminus\set{z^{\ell}})$ that $d_2(z)=\frac{p}{q}$ for some integer $p$, it follows that $d_2(z^{\ell})=\frac{p'}{q}$ for some integer $p'$ (since $q$ is an integer). 
This implies that $d_2$ is a $q$-rounded distribution. We also have $|d_1(z)-d_2(z)|<\frac{1}{q}$ for all $z\in Z$ (by ($\ddag$)) and thus all the desired properties have been established.
\end{itemize}
This completes the proof.
\end{proof}

\begin{corollary}\label{cor:q-rounded strategies}
For all almost-sure ergodic CMPGs, for all $\epsilon>0$, there exists an $\epsilon$-optimal, 
$q'$-rounded strategy $\sigma_1$ for Player~1, for all integers $q'\geq q$, where \[q=4\cdot \epsilon^{-1}\cdot m \cdot n^2\cdot (\delta_{\min})^{-\rand}\enspace .\]
\end{corollary}

\begin{proof}
Notice that the $q$ defined here is the same $q$ as is defined in Lemma~\ref{lemm:approx optimal strategy}. 
Let the integer $q'\geq q$ be given. Consider an almost-sure ergodic CMPG $G$. 
Let $\sigma_1'$ be a optimal stationary strategy in $G$ for Player~1. For each state $s$, pick a $q'$-rounded distribution $d^s$ over $\Gamma_1(s)$, such that $|\sigma_1'(s)(a_1)-d^s(a_1)|<\frac{1}{q'}\leq \frac{1}{q}$ for all $a_1\in \Gamma_1(s)$. Such a distribution exists by Lemma~\ref{lemm:round distribution}, since $q'\geq q\geq m\geq |\Gamma_1(s)|$. Let the strategy $\sigma_1$ be defined as follows: $\sigma_1(s)=d^s$ for each state $s\in S$. Hence $\sigma_1$ is a $q'$-rounded strategy. By Lemma~\ref{lemm:approx optimal strategy}, the strategy $\sigma_1$ is also an $\epsilon$-optimal strategy.
\end{proof}

\smallskip\noindent{\bf Exponential lower bound on patience.}
We now present a family of ergodic CMPGs where the lower bound
on patience is exponential in $\rand$.
We present the lower bound on a special class of ergodic CMPGs,
namely, skew-symmetric ergodic CMPGs which we define below.

\smallskip\noindent{\em Skew-symmetric CMPGs.}
A CMPG $G$ is {\em skew-symmetric}\footnote{For the special case of matrix games (that is; the case where $n=1$), this definition of skew-symmetry exactly corresponds to the notion of skew-symmetry for such.}, 
if there is a  bijective map $f:S \to S$, where $f(f(s))=s$, (for all $s$ we use $\ov{s}$ to denote $f(s)$) 
where the following holds: 
For each state $s$, there is a  bijective map $f_1^s:\Gamma_1(s)\to \Gamma_2(\ov{s})$ 
(for all $i\in \Gamma_1(s)$ we use $\ov{i}$ to denote $f_1^s(i)$) and a bijective map $f_2^s:\Gamma_2(s)\to\Gamma_1(\ov{s})$ (similarly to the first map, for all $j\in \Gamma_2(s)$ we use $\ov{j}$ to denote $f_2^s(j)$), 
such that for all $i\in \Gamma_1(s)$ and all $j\in \Gamma_2(s)$, the following 
conditions hold: (1)~we have $\cost(s,i,j)=1-\cost(\ov{s},\ov{j},\ov{i})$; 
(2)~for all $s'$ such that $\trans(s,i,j)(s')>0$, we have $\trans(\ov{s},\ov{j},\ov{i})(\ov{s}')=\trans(s,i,j)(s')$; 
and (3)~we have $f_2^{\ov{s}}(f_1^s(i))=i$ and that $f_1^{\ov{s}}(f_2^s(j))=j$.

\begin{lemma}\label{lemm:skew}
Consider a skew-symmetric CMPG $G$. Then for all $s$ we have $v_s=1-v_{\ov{s}}$.
\end{lemma}
\begin{proof}
Let $s$ be a state.
For a stationary strategy $\sigma_k$ for Player~$k$, $k \in \set{1,2}$, 
let $\ov{\sigma}_k$ be a stationary strategy for the other player defined as follows: For each state $s$ and action $i\in \Gamma_k(s)$, 
let $\ov{\sigma}_k(\ov{s})(\ov{i})=\sigma_k(s)(i)$.
For a stationary strategy $\sigma_1$ for Player~1, consider the stationary strategy profile $(\sigma_1,\ov{\sigma}_1)$.
For the random walk $P=\pat_s^{\sigma_1,\ov{\sigma}_1}$, 
where the players follows $(\sigma_1,\ov{\sigma}_1)$, starting in $s$ corresponds to the random walk 
$\ov{P}=\pat_{\ov{s}}^{\sigma_1,\ov{\sigma}_1}$, 
where the players follows $(\sigma_1,\ov{\sigma}_1)$, starting in $\ov{s}$, 
in the obvious way (that is: if $P$ is in state $s_i$ in the $i$-th step and the reward is $\lambda$, 
then $P'$ is in $\ov{s}_i$, in the $i$-th step and the reward is $1-\lambda$). The two random walks, $P$ and $P'$, are equally likely. 
This implies that $v_s=1-v_{\ov{s}}$.
\end{proof}

\begin{corollary}\label{coro:skew}
For all skew-symmetric ergodic CMPGs the value is $\frac{1}{2}$.
\end{corollary}

\noindent{\em Family $G^k_{\eta}$.}
We  now provide a lower bound for patience of $\epsilon$-optimal strategies in skew-symmetric 
ergodic CMPGs. 
More precisely, we give a family of games $\set{G^k_{\eta} \mid k\geq 2\vee 0<\eta<\frac{1}{4\cdot k+4}}$, 
such that $G^k_{\eta}$ consists of $2\cdot k+5$ states and such that $\delta_{\min}$ for $G^k_{\eta}$ is $\eta$. 
The game $G^k_{\eta}$ is such that all $\frac{1}{48}$-optimal stationary strategies require patience at least 
$\frac{1}{2\cdot \eta^{k/2}}$. 

\smallskip\noindent{\em Construction of the family $G^k_{\eta}$.}
For a given $k\geq 2$ and $\eta$, such that $0<\eta<\frac{1}{4\cdot k+4}$, let the game $G_{\eta}^k$ be as follows: 
The game consists of $2\cdot k+5$ states, $S=\set{a,b,\ov{b},c,\ov{c}, s_1,\ov{s}_1,s_2,\ov{s}_2,\dots,s_k,\ov{s}_k}$. 
For $s\in (S\setminus\set{c,\ov{c}})$, we have that $|\Gamma_1(s)|=|\Gamma_2(s)|=1$. 
For $s'\in \set{c,\ov{c}}$, we have that 
$|\Gamma_1(s')|=|\Gamma_2(s')|=2$, and let $\Gamma_1(s')=\set{i^{s'}_1,i^{s'}_2}$ 
and $\Gamma_2(s')=\set{j^{s'}_1,j^{s'}_2}$.
For $y\geq 2$ we have that $s_y$ (resp. $\ov{s}_y$) has a transition to $s_k$ (resp. $\ov{s}_k$) of probability $1-\eta$;  
to $s_{y-1}$ (resp. $\ov{s}_{y-1}$), where $s_0=\ov{s}_0=a$, with probability $\eta$; and also the reward of the transition is~$0$ (resp.~$1$).  
The state $b$ (resp. $\ov{b}$) is deterministic and has a transition to $a$ of reward~$0$ (resp.~$1$). 
The transition function at state $c$ is deterministic, and thus for each pair $(i,j)$ of actions we define the unique successor of $c$. 
\begin{enumerate}
\item For $(i^c_1,j^c_1)$ and $(i^c_2,j^c_2)$  the successor is $\ov{b}$. 
\item For $(i^c_1,j^c_2)$ the successor is $b$. 
\item For $(i^c_2,j^c_1)$ the successor is $s_k$. 
\end{enumerate}
The reward of the transitions from $c$ is~$0$. 
Intuitively, the transitions and rewards from $\ov{c}$ are defined from skew-symmetry.
Formally, we have: 
\begin{enumerate}
\item For $(i^{\ov{c}}_1,j^{\ov{c}}_1)$ and $(i^{\ov{c}}_2,j^{\ov{c}}_2)$  the successor is $b$. 
\item For $(i^{\ov{c}}_2,j^{\ov{c}}_1)$ the successor is $\ov{b}$. 
\item For $(i^{\ov{c}}_1,j^{\ov{c}}_2)$ the successor is $\ov{s}_k$. 
\end{enumerate}
The reward of the transitions from $\ov{c}$ is~$1$. 
There is a transition from $a$ to each other state. The probability to go to $c$ and the probability to go to $\ov{c}$ are both $\frac{1}{4}$. 
For each other state $s'$ (other than $c$, $\ov{c}$ and $a$), the probability to go to $s'$ from $a$ is $\frac{1}{4\cdot k+4}$. 
The transitions from $a$ have reward $\frac{1}{2}$. There is an illustration of $G^k_{\eta}$ in Figure~\ref{fig:ergodic skew-symmetric}.

\begin{lemma}
For any given $k$ and $\eta$, such that $0<\eta<\frac{1}{4\cdot k+4}$, 
the CMPG $G^k_{\eta}$ is both skew-symmetric and ergodic.
Thus $G^k_{\eta}$ has value $\frac{1}{2}$.
\end{lemma}
\begin{proof}
We first argue about ergodicity: from any starting state $s$, the state $a$ is reached almost-surely; and from $a$ there is a transition to all other states with positive probability. 
This ensures that $G^k_{\eta}$ is ergodic.

The following mappings implies that CMPG $G^k_{\eta}$ is skew-symmetric:
(i)~$f(s_i)=\ov{s}_i$ for all $i$; and (ii)~$f(a)=a$; and 
(iii)~$f(b)=\ov{b}$; and (iv)~$f(c)=\ov{c}$. 
The bijective map $f_1^c$ between $\Gamma_1(c)$ and $\Gamma_2(\ov{c})$ is such that $\ov{i}^c_1=j^{\ov{c}}_1$ (and thus also $\ov{i}^c_2=j^{\ov{c}}_2$). 
The bijective map $f^c_2$ is such that $\ov{j}^c_1=i^{\ov{c}}_1$ (and thus also $\ov{j}^c_2=i^{\ov{c}}_2$). 
\end{proof}

\newcommand{\C}{{\mathcal C}}

\begin{lemma}\label{lemm:low-bound1}
For any given $k$ and $\eta$, such that $0<\eta<\frac{1}{4\cdot k+4}$,
consider the set $\C_p$ of stationary strategies for Player~1 in $G^k_{\eta}$, 
with patience at most $\frac{1}{p}$, where $p=2\cdot \eta^{k/2}$.
Consider the stationary strategy $\sigma_1^*$ defined as:
(i)~$\sigma_1^*(c)(i^c_2)=p$ (and $\sigma_1^*(c)(i^c_1)=1-p$);
and
(ii)~ $\sigma_1^*(\ov{c})(i^{\ov{c}}_2)=1-p$ (and $\sigma_1^*(\ov{c})(i^{\ov{c}}_1)=p$).
Then the strategy $\sigma_1^*$ ensures the maximal value among all strategies
in $\C_p$.
\end{lemma}
\begin{proof}
First, observe that from $s_k$, the probability to reach $a$ in $k$ steps is $\eta^k$. 
If $a$ is not reached in $k$ steps, then in these $k$ steps $s_k$ is reached again. 
Similarly for $\ov{s}_k$. 
Thus, the expected length $L_{s_k}$ of a run from $s_k$ (or $\ov{s}_k$) 
to $a$, is (strictly) more than $\eta^{-k}$, but (strictly) less\footnote{It is also less than $2\cdot \eta^{-k}+k$, since for any state $s_i$, for $i\geq 1$, there is a probability of more than $\frac{1}{2}$ to go to $s_k$ and whenever the play is in $s_k$ there is a probability of $\eta^{k}$ that it is the last time.} than $k\cdot \eta^{-k}$. 

The proof is split in three parts. 
The first part considers strategies in $\C_p$ that plays $i^c_2$ with probability 
greater than $p$; the second part considers strategies in $\C_p$ that plays 
$i^c_2$ with probability~0; and the third part shows that the optimal distribution 
for the actions in $\ov{c}$ is to play as $\sigma_1^*$.

\begin{enumerate}

\item Consider some stationary strategy $\sigma_1' \in \C_p$ such that $\sigma_1'(c)(i^c_2)=p'>p$.
Consider the strategy $\sigma_1$ such that $\sigma_1(c)=\sigma_1^*(c)$ and $\sigma_1(\ov{c})=\sigma_1'(\ov{c})$. 
We show that $\sigma_1$ guarantees a higher expected mean-payoff value for the run between 
$a$ and $c$ than $\sigma_1'$, and thus $\sigma_1$ ensures 
greater mean-payoff value than $\sigma_1'$.  

For $\ell\in\set{1,2}$, let $\sigma_2^{\ell}$ be an arbitrary stationary strategy which plays $j^c_{\ell}$ with probability~1. 
Let $m_{\ell}$ be the mean-payoff of the run from $c$ to $a$, when Player~1 plays $\sigma_1^*$ and 
Player~2 plays $\sigma_2^{\ell}$.  Define $m'_{\ell}$ similarly, except that Player~1 plays $\sigma_1'$ 
instead of $\sigma_1^*$. 
Then, $m_1=\frac{1-p}{p\cdot (L_{s_k}+1)+(1-p)\cdot 2}$ and $m'_1=\frac{1-p'}{p'\cdot (L_{s_k}+1)+(1-p')\cdot 2}$ 
(the expected length of the run is $p'\cdot (L_{s_k}+1)+(1-p')\cdot 2$ and it gets reward~1 only once and only 
with probability $1-p'$). We now argue that $m_1>m'_1$. Consider $m_1-m'_1$:

\begin{align*}
m_1-m'_1
&
= \frac{1-p}{p\cdot (L_{s_k}+1)+(1-p)\cdot 2}-\frac{1-p'}{p'\cdot (L_{s_k}+1)+(1-p')\cdot 2} \\[2ex] 
&
=\frac{(1-p)\cdot (p'\cdot (L_{s_k}+1)+(1-p')\cdot 2)
- (1-p')\cdot (p\cdot (L_{s_k}+1)+(1-p)\cdot 2)}{(p\cdot (L_{s_k}+1)+(1-p)\cdot 2)\cdot (p'\cdot (L_{s_k}+1)+(1-p')\cdot 2)}
\end{align*}
Hence, see that the numerator of the above expression is
\begin{align*}
(1-p)\cdot (p'\cdot (L_{s_k}+1)+&(1-p')\cdot 2)
- (1-p')\cdot (p\cdot (L_{s_k}+1)+(1-p)\cdot 2) \\
&=
(p'-p)\cdot (L_{s_k}+1) >0
\end{align*}
and therefore $m_1>m'_1$.

We now argue that $m_1<m_2$ and $m_1'<m_2'$ (and thus Player~2 plays $j^c_1$ in $c$ against both $\sigma_1$ (and thus also $\sigma_1^*$) and $\sigma_1'$). 
We  have that $m_2=\frac{p}{2}$ and (repeated for convenience) $m_1=\frac{1-p}{p\cdot (L_{s_k}+1)+(1-p)\cdot 2}<\frac{1}{p\cdot (L_{s_k}+1)}<\frac{1}{p\cdot \eta^{-k}}$. 
But $p=2\cdot \eta^{k/2}$ and therefore $\frac{1}{p\cdot \eta^{-k}}\leq \frac{1}{2\cdot \eta^{k/2}\cdot \eta^{-k}}=\frac{1}{2\cdot \eta^{-k/2}}< \frac{1}{\eta^{-k/2}}\leq \frac{p}{2}$. 
Similar for $m_1'<m_2'$, and hence we have the desired result.

\item Consider some stationary strategy $\sigma_1^0 \in \C_p$ such that $\sigma_1^0(c)(i^c_2)=0$.
Now consider the strategy $\sigma_1$ such that $\sigma_1(c)=\sigma_1^*(c)$ and 
$\sigma_1(\ov{c})=\sigma_1^0(\ov{c})$. 
Then, the best response $\sigma_2^0$ for Player~2 against $\sigma_1^0$ 
plays $j^c_2$ with probability 1. 
We see that if Player~1 follows $\sigma_1^0$ and Player~2 follows $\sigma_2^0$, then 
the mean-payoff of the run from $c$ to $a$ is 0.
Thus $\sigma_1$ ensures greater mean-payoff value than $\sigma_1^0$.

\item Similar to the first two parts, it follows that a strategy that plays 
like $\sigma_1^*$ in $\ov{c}$ ensures at least the mean-payoff value of any
other stationary strategy in $\C_p$ for the play between $\ov{c}$ and $a$.
(In this case, the best response for Player~2 plays $j^{\ov{c}}_1$ with probability~$1$ 
and therefore the mean-payoff for the run from $\ov{c}$ to $a$ is 
$\frac{2-p}{2}$ as the length of the run is 2; and  
with probability $1-p$ both rewards are 1, otherwise the first reward is~$1$ 
and the second reward is $0$).

\end{enumerate}
It follows from above that $\sigma_1^*$ ensures the maximal mean-payoff value 
among all strategies in $\C_p$.
\end{proof}

\begin{lemma}\label{lemm:low-bound2}
For any given $k$ and $\eta$, such that $0<\eta<\frac{1}{4\cdot k+4}$,
consider the set $\C_p$ of stationary strategies for Player~1 in $G^k_{\eta}$, 
with patience at most $\frac{1}{p}$, where $p=2\cdot \eta^{k/2}$.
For all strategies in $\C_p$, the mean-payoff value is at most $\frac{23}{48}$; and 
hence no strategy in $\C_p$ is $\frac{1}{48}$-optimal.
\end{lemma}
\begin{proof}
By Lemma~\ref{lemm:low-bound1} we only need to consider $\sigma_1^*$ as defined
in Lemma~\ref{lemm:low-bound1}.
Now we calculate the expected mean-payoff value for a run from $a$ to $a$ 
given $\sigma_1^*$ and a positional best-response strategy $\sigma_2$ for Player~2, 
(which is then the expected mean-payoff value of the strategies in $G^k_{\eta}$) as follows:
\begin{enumerate}
\item With probability~$\frac{1}{2}$ in the first step, the run goes to some state which is neither $c$ nor $\ov{c}$. 
Since the probability is equally large to go to some state $s$ or to the corresponding skew-symmetric state $\ov{s}$ and no state $s$ can be reached such that $|\Gamma_1(s)|$ or $|\Gamma_2(s)|$ is more than 1, 
such runs has mean-payoff $\frac{1}{2}$. 
\item Otherwise with probability~$\frac{1}{2}$ in the first step we get reward $\frac{1}{2}$ and go to either $c$ or $\ov{c}$ with equal probability 
(that is: the probability to go to $c$ or $\ov{c}$ is $\frac{1}{4}$ each). 
As shown in Lemma~\ref{lemm:low-bound1}, 
(i)~the length of the run from $\ov{c}$ to $a$ is~2; and with probability $1-p$ both rewards are~1, 
otherwise the first reward is~$1$ and the second reward is $0$;
(ii)~the expected length of the run from $c$ to $a$ is $p\cdot (L_{s_k}+1)+(1-p)\cdot 2$ 
and it gets reward~1 only once and only with probability $1-p$ (where $L_{s_k}$ is as 
defined in Lemma~\ref{lemm:low-bound1}).
\end{enumerate}
From the above case analysis we conclude that the mean-payoff of the run from $a$ to $a$ is 
\begin{align*}
&\frac{1}{2}\cdot \frac{1}{2}+\frac{1}{4}\cdot \left(\frac{\frac{1}{2}+1+(1-p)}{3} +\frac{\frac{1}{2}+1-p}{1+p\cdot (L_{s_k}+1)+(1-p)\cdot 2}\right) \\[2ex]
&=\frac{1}{4}+\frac{\frac{1}{2}+1+(1-p)}{12} +\frac{\frac{1}{2}+1-p}{4\cdot (1+p\cdot (L_{s_k}+1)+(1-p)\cdot 2)} \\[2ex]
&<\frac{1}{4}+\frac{\frac{1}{2}+2}{12} +\frac{2}{4\cdot p\cdot L_{s_k}} 
\ \ < \ \ \frac{1}{4}+\frac{5}{24} +\frac{1}{2\cdot 2\cdot \eta^{k/2}\cdot \eta^{-k}} 
\\[2ex] &
=\frac{1}{4}+\frac{5}{24} +\frac{1}{4\cdot \eta^{-k/2}} 
\qquad 
< \ \ \frac{1}{4}+\frac{5}{24} +\frac{1}{48} 
\ \ = \ \ \frac{23}{48} 
\end{align*}
In the first inequality we use that $1+p\cdot (L_{s_k}+1)+(1-p)\cdot 2>p\cdot L_{s_k}$ and that $p>0$. In the second inequality we use that $p=2\cdot \eta^{k/2}$ and that $\eta^{-k}<L_{s_k}$. In the third we use that $\eta^{-k/2}>12$, which comes from $k\geq 2$ and $\eta<\frac{1}{4\cdot k+4}\leq \frac{1}{12}$. Therefore, we see that there is no $\frac{1}{48}$-optimal strategy with patience at most $\frac{\eta^{-k/2}}{2}$ in the game $G^k_{\eta}$.
\end{proof}

\begin{figure}
\begin{tikzpicture}[node distance=3cm]
\tikzstyle{every state}=[fill=white,draw=black,text=black,font=\small , inner sep=-0.05cm]

\matrix (v1) [label=below:$s_k$,minimum height=1.5em,minimum width=1.5em,matrix of math nodes,nodes in empty cells, left delimiter={.},right delimiter={.}]
{
\\
};
\draw[black] (v1-1-1.north west) -- (v1-1-1.north east);
\draw[black] (v1-1-1.south west) -- (v1-1-1.south east);
\draw[black] (v1-1-1.north west) -- (v1-1-1.south west);
\draw[black] (v1-1-1.north east) -- (v1-1-1.south east);

\matrix (v2) [label=right:$s_{k-1}$,right of= v1,minimum height=1.5em,minimum width=1.5em,matrix of math nodes,nodes in empty cells, left delimiter={.},right delimiter={.}]
{
\\
};
\draw[black] (v2-1-1.north west) -- (v2-1-1.north east);
\draw[black] (v2-1-1.south west) -- (v2-1-1.south east);
\draw[black] (v2-1-1.north west) -- (v2-1-1.south west);
\draw[black] (v2-1-1.north east) -- (v2-1-1.south east);

    \node[state,draw=white] (dots1) [right of=v2] {$\dots$};
    
\matrix (v3) [label=below:$s_1$,right of=dots1,node distance=1.5cm, minimum height=1.5em,minimum width=1.5em,matrix of math nodes,nodes in empty cells, left delimiter={.},right delimiter={.}]
{
\\
};
\draw[black] (v3-1-1.north west) -- (v3-1-1.north east);
\draw[black] (v3-1-1.south west) -- (v3-1-1.south east);
\draw[black] (v3-1-1.north west) -- (v3-1-1.south west);
\draw[black] (v3-1-1.north east) -- (v3-1-1.south east);

    \node[state,draw=white] (whitespace) [right of=v3] {};

\matrix (t) [label=below:$a$,below of =whitespace,node distance=1cm,minimum height=1.5em,minimum width=1.5em,matrix of math nodes,nodes in empty cells, left delimiter={.},right delimiter={.}]
{
\\
};
\draw[black] (t-1-1.north west) -- (t-1-1.north east);
\draw[black] (t-1-1.south west) -- (t-1-1.south east);
\draw[black] (t-1-1.north west) -- (t-1-1.south west);
\draw[black] (t-1-1.north east) -- (t-1-1.south east);

\matrix (vb1) [label=below:$\ov{s}_k$,below of=v1,minimum height=1.5em,minimum width=1.5em,matrix of math nodes,nodes in empty cells, left delimiter={.},right delimiter={.}]
{
\\
};
\draw[black] (vb1-1-1.north west) -- (vb1-1-1.north east);
\draw[black] (vb1-1-1.south west) -- (vb1-1-1.south east);
\draw[black] (vb1-1-1.north west) -- (vb1-1-1.south west);
\draw[black] (vb1-1-1.north east) -- (vb1-1-1.south east);

\matrix (vb2) [label=right:$\ov{s}_{k-1}$,below of= v2,minimum height=1.5em,minimum width=1.5em,matrix of math nodes,nodes in empty cells, left delimiter={.},right delimiter={.}]
{
\\
};
\draw[black] (vb2-1-1.north west) -- (vb2-1-1.north east);
\draw[black] (vb2-1-1.south west) -- (vb2-1-1.south east);
\draw[black] (vb2-1-1.north west) -- (vb2-1-1.south west);
\draw[black] (vb2-1-1.north east) -- (vb2-1-1.south east);

    \node[state,draw=white] (dots2) [below of=dots1] {$\dots$};
    
\matrix (vb3) [label=below:$\ov{s}_1$,below of=v3, minimum height=1.5em,minimum width=1.5em,matrix of math nodes,nodes in empty cells, left delimiter={.},right delimiter={.}]
{
\\
};
\draw[black] (vb3-1-1.north west) -- (vb3-1-1.north east);
\draw[black] (vb3-1-1.south west) -- (vb3-1-1.south east);
\draw[black] (vb3-1-1.north west) -- (vb3-1-1.south west);
\draw[black] (vb3-1-1.north east) -- (vb3-1-1.south east);

\matrix (s) [label=above:$b$,right of=v3, minimum height=1.5em,minimum width=1.5em,matrix of math nodes,nodes in empty cells, left delimiter={.},right delimiter={.}]
{
\\
};
\draw[black] (s-1-1.north west) -- (s-1-1.north east);
\draw[black] (s-1-1.south west) -- (s-1-1.south east);
\draw[black] (s-1-1.north west) -- (s-1-1.south west);
\draw[black] (s-1-1.north east) -- (s-1-1.south east);

\matrix (sb) [label=above:$\ov{b}$,below of=s, minimum height=1.5em,minimum width=1.5em,matrix of math nodes,nodes in empty cells, left delimiter={.},right delimiter={.}]
{
\\
};
\draw[black] (sb-1-1.north west) -- (sb-1-1.north east);
\draw[black] (sb-1-1.south west) -- (sb-1-1.south east);
\draw[black] (sb-1-1.north west) -- (sb-1-1.south west);
\draw[black] (sb-1-1.north east) -- (sb-1-1.south east);

\matrix (m) [label=right:$c$,right of=s,minimum height=1.5em,minimum width=1.5em,matrix of math nodes,nodes in empty cells, left delimiter={.},right delimiter={.}]
{
&\\
&\\
};
\draw[black] (m-1-1.north west) -- (m-1-2.north east);
\draw[black] (m-1-1.south west) -- (m-1-2.south east);
\draw[black] (m-2-1.south west) -- (m-2-2.south east);
\draw[black] (m-1-1.north west) -- (m-2-1.south west);
\draw[black] (m-1-2.north west) -- (m-2-2.south west);
\draw[black] (m-1-2.north east) -- (m-2-2.south east);

\matrix (mb) [label=right:$\ov{c}$,right of=sb,minimum height=1.5em,minimum width=1.5em,matrix of math nodes,nodes in empty cells, left delimiter={.},right delimiter={.}]
{
&\\
&\\
};
\draw[black] (mb-1-1.north west) -- (mb-1-2.north east);
\draw[black] (mb-1-1.south west) -- (mb-1-2.south east);
\draw[black] (mb-2-1.south west) -- (mb-2-2.south east);
\draw[black] (mb-1-1.north west) -- (mb-2-1.south west);
\draw[black] (mb-1-2.north west) -- (mb-2-2.south west);
\draw[black] (mb-1-2.north east) -- (mb-2-2.south east);

 \draw[->](m-1-1.center) .. controls +(180:1) ..  (sb-1-1.east);
 \draw[->](m-2-2.center) .. controls +(-90:1) ..  (sb-1-1.south east);
 
  \draw[->](m-2-1.center) .. controls +(180:1) ..  (s-1-1.east);
    \draw[->](m-1-2.center) .. controls +(120:2) ..  (v1-1-1.north east);
    
 \draw[->,dashed](mb-2-1.center) .. controls +(180:1) ..  (s-1-1.south east);
 \draw[->,dashed](mb-1-2.center) .. controls +(90:1) ..  (s-1-1.east);
 
  \draw[->,dashed](mb-1-1.center) .. controls +(180:1) ..  (sb-1-1.east);
    \draw[->,dashed](mb-2-2.center) .. controls +(-120:2) ..  (vb1-1-1.south east);

 \draw[->](v1-1-1.center) .. controls +(0:1) .. node[near end,below] (x) {$\eta$}  (v2-1-1.west);
\draw[->](v1-1-1.center)  .. controls +(45:1) and +(90:1).. node[midway,above] (x) {$1-\eta$} (v1-1-1.north);
 \draw[->](v2-1-1.center) .. controls +(-90:1).. node[ near end,below] (x) {$1-\eta$}  (v1-1-1.south east);
  \draw[->](v2-1-1.center) .. controls +(-90:1).. node[ near end,below] (x) {$\eta$}  (dots1);
  
   \draw[->](v3-1-1.center) .. controls +(90:1).. node[ near end,below] (x) {$1-\eta$}  (v1-1-1.north east);
  \draw[->](v3-1-1.center) .. controls +(90:1).. node[ near end,above] (x) {$\eta$}  (t);
  
   \draw[->,dashed](vb1-1-1.center) .. controls +(0:1) .. node[near end,below] (x) {$\eta$}  (vb2-1-1.west);
\draw[->,dashed](vb1-1-1.center)  .. controls +(45:1) and +(90:1).. node[midway,above] (x) {$1-\eta$} (vb1-1-1.north);
 \draw[->,dashed](vb2-1-1.center) .. controls +(-90:1).. node[ near end,below] (x) {$1-\eta$}  (vb1-1-1.south east);
  \draw[->,dashed](vb2-1-1.center) .. controls +(-90:1).. node[ near end,below] (x) {$\eta$}  (dots2);
  
   \draw[->,dashed](vb3-1-1.center) .. controls +(90:1).. node[ near end,above] (x) {$1-\eta$}  (vb1-1-1.north east);
  \draw[->,dashed](vb3-1-1.center) .. controls +(90:1).. node[ near end,below right] (x) {$\eta$}  (t);
    \draw[->](s-1-1.center) .. controls +(-90:0.1)..   (t.north);
         \draw[->,dashed](sb-1-1.center) .. controls +(45:0.95)..   (t.south east);
\end{tikzpicture}
\caption{For a given $k\in\N$ and $0<\eta<\frac{1}{4\cdot k+4}$, the skew-symmetric ergodic game $G^k_{\eta}$, except that the transitions from $a$ are not drawn. There is a transition from $a$ to each other state. The probability to go to $c$ from $a$ and the probability to go to $\ov{c}$ from $a$ are both $\frac{1}{4}$. For each other state $s'$ (other than $c$, $\ov{c}$ and $a$), the probability to go to $s'$ from $a$ is $\frac{1}{4\cdot k+4}$. The transition from $a$ has reward $\frac{1}{2}$. Dashed edges have reward $1$ and non-dashed edges have reward $0$. 
Actions are annotated with probabilities if the successor is not deterministic.\label{fig:ergodic skew-symmetric}}
\end{figure}

\begin{theorem}[Strategy complexity]
The following assertions hold:
\begin{enumerate}
\item \emph{(Upper bound).} 
For almost-sure ergodic CMPGs, for all $\epsilon>0$, 
there exists an $\epsilon$-optimal strategy of patience at most $\lceil{4\cdot \epsilon^{-1}\cdot m \cdot n^2\cdot (\delta_{\min})^{-\rand}\rceil}$.
\item \emph{(Lower bound).}
There exists a family of ergodic CMPGs $G_n^{\delta_{\min}}$, for each odd $n\geq 9$ and $0<\delta_{\min}<\frac{1}{2\cdot n}$ and $n=\rand +5$, 
such that any $\frac{1}{48}$-optimal strategy in $G_n^{\delta_{\min}}$ has patience at least $\frac{1}{2} \cdot (\delta_{\min})^{-\rand/4}$. 
\end{enumerate}
\end{theorem}
\begin{proof}
The upper bound comes from Corollary~\ref{cor:q-rounded strategies}, since all $q$-rounded strategies have patience at most $q$; and the lower
bound follows from Lemma~\ref{lemm:low-bound2}.
\end{proof}

\subsection{Hardness of approximation}\label{subsec:ssg hard}
We present a polynomial reduction from the value problem for SSGs to the 
problem of approximation of values for turn-based stochastic ergodic 
mean-payoff games (TEMPGs).

\smallskip\noindent{\bf The reduction.}
Consider an SSG $G$ with $n$ non-terminal states, and two terminal states ($\top$ and $\bot$).
Given a state $s$ in $G$ we construct a TEMPG $G'=\Red(G,s)$ that has the same states as $G$ 
(including the terminal states) and one additional state $s'$. 
For every transition in $G$, there is a corresponding transition in $G'$, with reward~0. 
The 1~terminal $\top$ (resp. 0~terminal $\bot$) instead of the self-loop, has two outgoing 
transitions that go to $\top$ (resp. $\bot$) with probability $1-\frac{1}{2^{9n}}$
and to $s'$ with probability $\frac{1}{2^{9n}}$. 
The reward of the transitions are~$1$ (resp.~$0$) for $\top$ (resp. $\bot$). 
The additional state $s'$ goes to $s$ with probability $1-\frac{1}{2^{7n}}$ 
and to each other state (including the terminals, but not $s$ and $s'$) 
with probability $\frac{1}{(n+1)\cdot 2^{7n}}$. 
The rewards of the transitions from $s'$ are $0$. 
We first observe that the game $G'$ is ergodic: 
since the SSG $G$ is stopping, from all states and for all strategies in $G$, 
the terminal states are reached with probability~1; and hence in $G'$, 
from all states and for all strategies, the state $s'$ is reached with probability~1;
and from $s'$ there exists a positive transition probability to every state other than
$s'$.
It follows that under all strategy profiles, from all starting states, the state $s'$ is 
visited infinitely often almost-surely, 
and hence every other state is visited infinitely often almost-surely.
Hence $G'$ is ergodic. 
We now show that the value $v$ of $G'$ is ``close'' to the value $v_s$ of $s$ in $G$. 
We then argue that we can obtain $v_s$ from $v$ in polynomial time by rounding.

\begin{lemma}\label{lemm:reduction}
Let $G$ be an SSG, and  consider a state $s$ in $G$ with value $v_s$. 
The value $v$ of $\Red(G,s)$ is in the interval $[v_s-2^{-7n+1};v_s+2^{-7n+1}]$.
\end{lemma}

\begin{proof}
We show that the value of $G'$ is at least $v_s-2^{-7n+1}$; and the other part 
of the proof is symmetric.
Notice that since $G$ is stopping, we reach a terminal in $n$ steps with 
probability at least $\frac{1}{2^{n}}$, from every starting state. 
The expected number of steps required to reach the terminal states is at most
$n\cdot 2^n$ (one can also use a more refined argument similar to~\cite{ESA} 
to show that the expected number of steps is at most $2^{n+1}$). 
By construction this is also the case in $G'$. 
From a terminal state in $G'$ the expected number of steps required to reach $s'$ is $2^{9n}$.
Consider an optimal strategy $\sigma_1$ in $G$ for Player~1. Since $G$ and $G'$ have the same set of states 
where Player~1 has a choice (and the same choices in those states), we can also use $\sigma_1$ in $G'$. 
Now consider the best response strategy $\sigma_2$ against $\sigma_1$ for Player~2 in $G'$. 
We now estimate the value of $G'$. 
The best $\sigma_2$ can ensure for Player~2 is the following: 
\begin{itemize}
\item By the argument above, for the plays from any starting state in $G$, 
the expected number of steps required to reach a terminal state is (at most) 
$n \cdot 2^n$.
\item For a state $t$ different from $s$, the plays from $t$ reach the 0~terminal with probability~1.
\item The plays from $s$ reach the 0~terminal with probability~$1-v_s$ and the 1~terminal with probability~$v_s$.
\end{itemize}  
Notice that for plays starting from any state $t\neq s'$, the expected number of steps to reach $s'$ is at most 
$n\cdot 2^n+2^{9n}$.
Hence the expected number of steps required to reach $s'$ again from itself is at most $n\cdot 2^n+2^{9n}+1$.
We now argue that the mean-payoff value is at least $v_s-2^{-7n+1}$.
With probability $1-\frac{1}{2^{7n}}$, the successor of $s'$ is $s$.
From $s$ the play reaches $s'$ after being in the 1~terminal for $v_s\cdot 2^{9n}$ steps in expectation. 
Each reward obtained in the 1~terminal is~1. All remaining rewards are~0. 
Hence, the mean-payoff value is at least 
\begin{align*}
\frac{v_s\cdot 2^{9n}\cdot (1-\frac{1}{2^{7n}})}{n\cdot 2^n+2^{9n}+1} &=\frac{v_s\cdot 2^{9n}}{n\cdot 2^n+2^{9n}+1}-\frac{v_s\cdot 2^{9n} \cdot \frac{1}{2^{7n}}}{n\cdot 2^n+2^{9n}+1}\\[2ex]
&\geq \frac{v_s\cdot 2^{9n}}{(1+2^{-7n})2^{9n}}-\frac{v_s\cdot 2^{9n}\cdot 2^{-7n}}{2^{9n}}\\[2ex]
&> (1-2^{-7n})\cdot v_s-v_s\cdot 2^{-7n}\\[2ex]
&= v_s-v_s\cdot 2^{-7n+1}\\[2ex]
&\geq v_s-2^{-7n+1} \enspace .
\end{align*}
The first inequality comes from $n\cdot 2^n=2^{n+\log n}<2^{2n}$;
the second inequality comes from $1-2^{-14n}=(1-2^{-7n})(1+2^{-7n})<1\Rightarrow 1-2^{-7n}<\frac{1}{1+2^{-7n}}$;
and the last inequality comes from $v_s\leq 1$.

Using a similar argument for Player~2, we obtain that the mean-payoff value is at most $v_s+2^{-7n+1}$, 
by using that the expected path-length from a state $t$ in $G$ to a terminal is at least 0. 
Therefore $v$, the value of $G'$, is in the interval $[v_s-2^{-7n+1};v_s+2^{-7n+1}]$.
\end{proof}

Observe that if the value $v$ of $G'$ can be approximated within $2^{-6n}$, 
then Lemma~\ref{lemm:reduction} implies that the approximation $a$ is 
in $[v_s-2^{-7n+1}-2^{-6n};v_s+2^{-7n+1}+2^{-6n}]$; which shows that $a$ is in 
$[v_s-2^{-5n};v_s+2^{-5n}]$. 
Hence we see that $a-2^{-5n}$ is in $[v_s-2^{-4n};v_s]$.
As observed by Ibsen-Jensen and Miltersen~\cite{ESA}, 
if the value of a state of an SSG can be approximated from below 
within $2^{-4n}$, then one can use the Kwek-Mehlhorn algorithm~\cite{kwek} 
to round the approximated value to obtain the correct value, 
in polynomial time. We therefore get the following lemma.

\begin{lemma}\label{lemm:ssg hard}
The problem of finding the value of a state in an SSG is polynomial time Turing reducible to the problem of 
approximating the value of a TEMPG (turn-based stochastic ergodic mean-payoff game) within $2^{-6n}$.
\end{lemma}

\subsection{Approximation complexity}\label{subsec:approx}
In this section we establish the approximation complexity 
for almost-sure ergodic CMPGs.
We first recall the definition of the decision problem for approximation.

\smallskip\noindent{\bf Approximation decision problem.}
Given an almost-sure ergodic CMPG $G$ (with rational transition probabilities given in 
binary), a state $s$, an $\epsilon>0$ (in binary), and a rational number $\lambda$ (in binary), 
the promise problem $\PromValExtErg$ (i)~accepts if the value of $s$ is at least $\lambda$, 
(ii)~rejects if the value of $s$ is at most $\lambda-\epsilon$, and 
(iii)~if the value is in the interval $(\lambda-\epsilon;\lambda)$, then it may both 
accept or reject. 


\begin{theorem}[Approximation complexity]
For almost-sure ergodic CMPGs, the following assertions hold:
\begin{enumerate}
\item \emph{(Upper bound).} The problem $\PromValExtErg$ is in $\FNP$. 
\item \emph{(Hardness).} The problem of finding the value of a state in an SSG is polynomial time Turing reducible to the problem $\PromValExtErg$,
even for the special case of turn-based stochastic ergodic mean-payoff games (TEMPGs).
\end{enumerate}
\end{theorem}
\begin{proof}
We present the proof for both the items.
\begin{enumerate}
\item 
We first present an $\FNP$ algorithm for $\PromValExtErg$ as follows: 
Guess an $\frac{\epsilon}{4}$-optimal, $q'$-rounded strategy $\sigma_1$ for Player~1, where $q'=\lceil q\rceil$ such that $q$ is as in Corollary~\ref{cor:q-rounded strategies}
(also such a strategy exists by Corollary~\ref{cor:q-rounded strategies}).  
The strategy is then described using at most $O(n\cdot m\cdot \log q')$ many bits. 
Since $\epsilon$ and $\delta_{\min}$ is given in binary, $\log q'$ uses at most polynomial many bits. 
Now compute the best response strategy for Player~2. Since $\sigma_1$ is a stationary strategy (because it is $q'$-rounded), 
when Player~1 restricted to follow $\sigma_1$, the game becomes an MDP for Player~2, and the size of the MDP 
is also polynomial in the size of $G$ and $\log q'$. 
Hence there exists a positional best response strategy $\sigma_2$, which we can find in polynomial time using linear programming~\cite{FV97,Puterman,Karmarkar}. 
When Player~1 follows $\sigma_1$ and Player~2 follows $\sigma_2$ some expected mean-payoff $\val$ is achieved.  
Similarly guess an $\frac{\epsilon}{4}$-optimal, $q$-rounded strategy $\sigma'_2$ for Player~2. Again there exists a 
positional best response strategy $\sigma'_1$ for Player~1 which can again be computed in polynomial time. 
When Player~1 follows $\sigma'_1$ and Player~2 follows 
$\sigma'_2$ some expected mean-payoff $\val'$ is achieved. 
If $\val'-\val>\frac{\epsilon}{2}$, then reject, because then not both $\sigma_1$ and $\sigma'_2$ 
can be $\frac{\epsilon}{4}$ optimal. 
Clearly the value of $G$ must be in $[\val;\val']$.  
Notice that both $\lambda-\epsilon$ and $\lambda$ cannot be in $[\val;\val']$, 
since $\val'-\val\leq \frac{\epsilon}{2}$. Therefore if $\lambda\leq \val'$, then accept, otherwise reject.
This establishes that $\PromValExtErg$ is in $\FNP$.

\item 
We now show that the problem of finding the value of a state in an SSG is polynomial time Turing reducible 
to the problem $\PromValExtErg$ for TEMPGs.  By Lemma~\ref{lemm:ssg hard}, we just need to approximate the value $v$ of a 
TEMPG $G$ within $2^{-6n}$. For any number $0<a<1$ and integer $b$, let $\Proc^b$ be a procedure, 
that takes $\frac{p}{q}$ as an input and returns if $a\geq \frac{p}{q}$, where $0\leq p\leq q\leq b$. 
For any integer $b$, given procedure $\Proc^b$, the Kwek-Mehlhorn algorithm~\cite{kwek}, 
finds integers $0\leq p\leq q\leq b$, such that $a-\frac{p}{q}< \frac{1}{b}$ in $O(\log b)$ time and 
$O(\log b)$ calls to $\Proc^b$. 
We argue how to use the Kwek-Mehlhorn algorithm~\cite{kwek} to find the value of $G$ within $2^{-6n}$ 
using polynomially many calls to $\PromValExtErg$. Let $b$ be $2^{8n}$. Let $\Proc^b$ be  
$\PromValExtErg$ with $\epsilon=2^{-16n}$. Notice that the choice of $\epsilon$ ensures that there can be 
at most one pair $p,q$ such that $\frac{p}{q}\in [v-\epsilon;v]$, where $0\leq p\leq q\leq 2^{8n}$, 
because all such numbers are at least $2^{-16n}$ apart. On such an input $\PromValExtErg$ answers arbitrarily, 
but on all other inputs it accurately answers if $\frac{p}{q}\geq v$. The Kwek-Mehlhorn algorithm  
queries a pair of variables only once, and finds a fraction $\frac{p}{q}$ 
such that $0\leq p\leq q\leq 2^{8n}$. But the four best such fractions must be within 
$2^{-6n}$ of $v$.
\end{enumerate}
The desired result follows.
\end{proof}

\section{Strategy-iteration Algorithm for Almost-sure Ergodic CMPGs}\label{sec:strategy iteration ergodic}
The classic algorithm for solving ergodic CMPGs was given by Hoffman and Karp~\cite{HK}.
We present a variant of the algorithm, and show that for every $\epsilon>0$ it runs in 
exponential time for $\epsilon$ approximation.
Also observe that even for the value problem for SSGs the strategy-iteration algorithms 
require exponential time~\cite{Fr11,Fe10}, and hence our exponential upper bound is optimal 
(given our reduction of the value problem of SSGs to the approximation problem for 
TEMPGs).

\smallskip\noindent{\bf The variant of Hoffman-Karp algorithm.} 
For an almost-sure ergodic CMPG $G$, an $\epsilon>0$, and a state $t$, 
we present an algorithm to compute a $q$-rounded $\epsilon$-optimal strategy in $O(q^{n\cdot m})$ iterations, 
and each iteration requires $O\big(2^{\poly(m)} \cdot \poly(n,\log(\epsilon^{-1}), \log(\delta_{\min}^{-1})) \big)$ 
time, 
where  \[q=\left\lceil 4\cdot \epsilon^{-1}\cdot m \cdot n^2\cdot (\delta_{\min})^{-\rand}\right\rceil\enspace .\]
Note that in all typical cases, $n$ is large and $m$ is constant, and every iteration takes 
polynomial time if $m$ is constant.
The basic informal description of the algorithm is as follows.
In every iteration $i$, the algorithm considers a $q$-rounded strategy $\sigma_1^i$, 
and then improves the strategy locally as follows: first it computes the \emph{potential} $v^{\sigma_1^i}_s$
given $\sigma_1^i$ as in the Hoffman-Karp algorithm, and then for every state $s$, 
the algorithm locally computes the best $q$-rounded distribution at $s$ to improve
the potential.
The intuitive description of the potential is as follows: 
Fix the specific state $t$ as a target state (where the potential must be~0); 
and given a stationary strategy $\sigma$, consider a modified reward function that assigns 
the original reward minus the value ensured by $\sigma$.
Then the potential for every state $s$ other than the specified state $t$ is the expected sum of rewards under 
the modified reward function for the random walk from $s$ to $t$.
The local improvement step is achieved by playing a matrix game with potentials.
Our variant differs from the Hoffman-Karp algorithm that while solving
the matrix game we restrict Player~1 to only $q$-rounded distributions.
The formal description of the algorithm is given in Figure~\ref{alg-main},
and the formal definition of the expected one-step reward $\ExpRew(s,d_1,d_2)$ for 
distributions $d_1$ over $\Gamma_1(s)$ and $d_2$ over $\Gamma_2(s)$ is as follows:
$\ExpRew(s,d_1,d_2)=\sum_{a_1 \in \Gamma_1(s), a_2 \in \Gamma_2(s)} \cost(s,a_1,a_2) \cdot d_1(a_1) \cdot d_2(a_2)$.

\begin{figure}
\begin{center}
\begin{function}[H]
Let $q\leftarrow \left\lceil 4\cdot \epsilon^{-1}\cdot m \cdot n^2\cdot (\delta_{\min})^{-\rand}\right\rceil$\;
Let $\sigma_1^0$ be a $q$-rounded strategy\;
\For{$(i\in \Z_+)$}{
Compute $g^i,(v_s^i)_{s\in S}$ as the unique solution of 
\begin{align*}
\forall s\in S: g^i+v^i_s & = \min_{a_2\in \mov_2(s)}(\ExpRew(s,\sigma_1^{i-1}(s),a_2)+\sum_{s'\in S}\trans(s,\sigma_1^{i-1}(s),a_2)(s') \cdot v_{s'}^i) \\
v^i_t & =0;
\end{align*}
\For{$(s\in S)$}
{
Let $M_s$ be the matrix game defined as follows:
$M_s[a_1,a_2]\leftarrow \cost(s,a_1,a_2) + \sum_{s'\in S}\trans(s,a_1,a_2)(s') \cdot v_{s'}^i$, for all $a_1\in \Gamma_1(s)$ and $a_2\in \Gamma_2(s)$\;
\If{{\em $(\sigma_1^{i-1}(s)$ is a best $q$ rounded distribution for the matrix game $M_s)$}}
{
Let $\sigma_1^i(s)\leftarrow \sigma^{i-1}_1(s)$\;
}
\Else
{
Let $\sigma_1^i(s)$ be an arbitrary best $q$-rounded distribution over $\Gamma_1(s)$ for the matrix game $M_s$\;
}
}
\If{$(\sigma_1^i= \sigma_1^{i-1})$}{\Return{$\sigma_1^i$}\;}
}
\caption{VarHoffmanKarp($G$,$\epsilon$,$t$)}
\end{function}

\end{center}
\vspace*{-0.5cm}
\caption{\label{alg-main}Algorithm for solving ergodic games}
\vspace*{-0.3cm}
\end{figure}

\smallskip\noindent{\bf Computation of every iteration.} 
The computation of every iteration is as follows. 
The computation of the unique solution $g^i$ and $(v^i_s)_{s\in S}$ is obtained 
in polynomial time using linear programming. 
The fact that the solution is unique follows from the fact that once a strategy for Player~1 is
fixed, we obtain an MDP for Player~$2$, and then the MDP solution is unique.
For a state $s$, let $\D^q(s)$ denote the set of all $q$-rounded distributions over $\Gamma_1(s)$.
A $q$-rounded distribution $d$ is \emph{best} for the matrix game $M_s$ iff 
$d\in \arg\max_{d_1 \in \D^q(s)} \min_{a_2 \in \Gamma_2(s)} \sum_{a_1 \in \Gamma_1(s)} d_1(a_1) \cdot M_s[a_1,a_2]$.
The computation of a best $q$-rounded strategy is achieved as follows:
given an $(m_1 \times m_2)$-matrix game $M$, solve the following integer linear program for 
$v$ and $(x_i)_{1\leq i \leq m_1}$:
\begin{align*}
 \max \ & v \\
\text{subject to } \quad & v  \leq \sum_{i=1}^{m_1} M[i,j] \cdot x_i; \qquad 1\leq j \leq m_2,\\
 & \sum_{i=1}^{m_1} x_i  =1; \\
& x_i\cdot q \in \N; \qquad 1 \leq i \leq m_1\\
& v\cdot q \cdot \ell \in \Z;
\end{align*}
where $\ell$ is the gcd of all the entries of $M$. 
It was shown by Lenstra~\cite{Lenstra}, that any integer linear programming problem on an integer $(m_1\times m_2)$-matrix 
(that is, with $m_1$ variables) can be solved in time $2^{\poly(m_1)}\cdot \poly(m_2,\log a)$, where $a$ is an upper 
bound on the greatest integer in the matrix and associated vectors. Notice that we can simply scale our matrix with $q\cdot \ell$ and 
obtain our optimization problem in the required form. Since the entries in the original game was defined from a 
solution to an MDP (which can be represented using polynomially many bits, because the Player-1 strategy is 
$q$-rounded), we know that only polynomially many bits are needed to represent $M_s$ (also after scaling). 
Thus, such an integer linear programming problem can be solved in time 
$O(2^{\poly(m)}\cdot \poly(n,\log (\epsilon^{-1}),\log (\delta_{\min}^{-1})))$.
This gives us the desired time bound for every iteration.


\smallskip\noindent{\bf Turn-based game for correctness.}
For the correctness analysis, we consider a turn-based stochastic version 
of the game (which is not ergodic), and refer to the turn-based game as $G'=\RR(G)$.
The game $G'=\RR(G)$ is a bipartite game of exponential size.
For a state $s$ in $G$, let $S^q_s=\{(s\times d_1)\mid d_1 \in \D^q_s\}$.
The state space in $G'$ is $S'= (\bigcup_{s\in S}S^q_s) \cup S$. 
Whenever we mention $S$ in the rest of this paragraph it should be clear 
from the context if we refer to $S$ as a part of $G$ or $G'$. 
In $G'$, Player~1 controls the states in $S$ and Player~2 the ones in $\bigcup_{s\in S}S^q_s$. 
From state $s\in S$, for every $d_1 \in \D_s^q$, there is a transition from $s$ to $(s,d_1)\in S^q_s$ with reward~0;
and from each state $(s,d_1)\in S^q_s$ there are $|\mov_2(s)|$ actions. 
For an action $a_2 \in \Gamma_2(s)$, the probability distribution over the next 
state is given by $\trans(s,d_1,a_2)$, and the reward is given by $\ExpRew(s,d_1,a_2)$.
Given a $q$-rounded strategy $\sigma_1$ for Player~1 in $G$ and a positional
strategy $\sigma_2$, if we interpret the strategies in $G'$, then the mean-payoff
value in $G'$ is exactly half of the mean-payoff value in $G$. 

\smallskip\noindent{\bf Correctness analysis and bound on iterations.}
We now present the correctness analysis, and the bound on the number of iterations follows.
The classic strategy-iteration algorithm computes the same series of strategies for Player~1 on $\RR(G)$ as our 
modified Hoffman-Karp algorithm does on the original game\footnote{The proof that the strategy-iteration algorithm works for turn-based mean-payoff games seems to be folk-lore, and also see~\cite{RCN} for the related class of discounted games.
Moreover, though $\RR(G)$ is not almost-sure ergodic, if we consider an ergodic component $C$ in $G$, and consider the corresponding set of states in $\RR(G)$, 
then from all states in $C$ every other state in $C$ is visited infinitely often with probability~1 in $\RR(G)$.
Thus for the concrete game $\RR(G)$, the proof can also be done similarly to the proof by Hoffman and Karp~\cite{HK} for ergodic games, 
by picking $t$ as a state in $C$ in $\RR(G)$.}. 
This is because, if we consider a fixed strategy for Player~1 in $\RR(G)$ and the corresponding strategy in $G$, 
then the best response positional strategy for Player~2 in $\RR(G)$ and $G$ resp. must correspond to each other. 
Then, by the way the potentials are calculated by the two algorithms, we get the same potential for a given state $s\in S$
for Player~1 in $\RR(G)$ as we do for the corresponding state in $G$ (they are precisely the same, since the value in $\RR(G)$, for any given strategy profile, is half the value of $G$, and thus, when we have taken two steps in $\RR(G)$ we have subtracted precisely the value of $G$). 
For $d_1 \in \D^q_s$, the potential of state $(s,d_1)$ in $\RR(G)$ is the same as the value ensured for 
Player~1 in $M_s$, if Player~1 plays $d_1$. 
Thus, also the next strategy for Player~1 is the same.
Thus, since the turn-based algorithm correctly finds the optimal strategy for Player~1, 
our modified Hoffman-Karp algorithm also correctly finds the $q$-rounded strategy 
that guarantees the highest value in $G$ for all states, among all $q$-rounded strategies. 
Since the best $q$-rounded strategy in $G$ is $\epsilon$-optimal for $G$ (by Corollary~\ref{cor:q-rounded strategies}), 
we have thus found an $\epsilon$-optimal strategy.
It is well known that the classic strategy-iteration algorithm only considers each strategy for Player~1 once (because the potential of the strategies picked by Player~1 are monotonically increasing in every iteration of the loop) . 
Therefore our $\VarHoffmanKarp$ algorithm requires at most $q^{m\cdot n}$ iterations, since there are most $q^{m\cdot n}$ strategies 
that are $q$-rounded.

\smallskip\noindent{\bf Inefficiency in reduction to $\RR(G)$.}
Observe that we only use $\RR(G)$ for the correctness analysis, and do not 
explicitly construct $\RR(G)$ in our algorithm.
Constructing $\RR(G)$ and then solving $\RR(G)$ using strategy iteration 
could also be used to compute $\epsilon$-optimal $q$-rounded strategies.
However, as compared to our algorithm there are two drawbacks in constructing 
$\RR(G)$ explicitly.
First, then every iteration would take time polynomial in $q$ (which is exponential in $n$), 
whereas every iteration of our algorithm requires only polynomial time in $n$ and $\log q$.
Second, our algorithm only requires polynomial space, whereas the construction 
of $\RR(G)$ would require space polynomial in $q$ (which is exponential in the input size).

\begin{theorem}
For an almost-sure ergodic CMPG, for all $\epsilon>0$, $\VarHoffmanKarp$ correctly computes an 
$\epsilon$-optimal strategy, and 
(i)~requires at most  
$O\left(\big(\epsilon^{-1}\cdot m \cdot n^2\cdot (\delta_{\min})^{-\rand}\big)^{n \cdot m}\right)$ 
iterations, and each iteration requires at most $O(2^{\poly(m)}\cdot \poly(n,\log (\epsilon^{-1}),\log (\delta_{\min}^{-1})))$ 
time; and
(ii)~requires polynomial space. 
\end{theorem}

\section{Analysis of the Value-iteration Algorithm}
In this section we show that the classical value-iteration algorithm 
requires at most exponentially many steps to approximate the value 
of ergodic concurrent mean-payoff games (ECMPGs).
We first start with a few notations and a basic lemma.

\smallskip\noindent{\bf Notations.}
Given an ECMPG $G$, let $v^*$ denote the value of the game (recall that all 
states in an ECMPG have the same value).
Let $v_s^T =\sup_{\stra_1 \in \bigstra_1 } \inf_{\stra_2 \in \bigstra_2} 
\Exp_s^{\stra_1,\stra_2}[\Avg_T]$ denote the value function for the objective
$\Avg_T$, i.e., playing the game for $T$ steps.
For an ECMPG $G$ we call the game with the objective $\Avg_T$ as $G_T$.
A \emph{Markov} strategy only depends on the length of the play and the current state.
A strategy $\stra_1$ is optimal for the objective $\Avg_T$ if
$v_s^T =\inf_{\stra_2 \in \bigstra_2} \Exp_s^{\stra_1,\stra_2}[\Avg_T]$,
and a strategy $\stra_2$ is optimal for the objective $\Avg_T$ if 
$v_s^T =\sup_{\stra_1 \in \bigstra_1 } \Exp_s^{\stra_1,\stra_2}[\Avg_T]$.
For the objective $\Avg_T$, optimal Markov strategies exist for both the players. 
The function $v_s^T$ is computed iteratively in $T$: initially $v^0_s=0$ for all $s$, 
and in every iteration $j \geq 1$ compute the following one-step operator for all $s$:
consider a matrix $M_s^j$ such that for all $a_1 \in \mov_1(s)$ and 
$a_2\in \mov_2(s)$ we have 
\[
M_s^j(a_1,a_2)= \frac{1}{j}\cdot  \left( \cost(s,a_1,a_2) + (j-1)\cdot \sum_{t \in S} v^{j-1}_t \cdot \trans(s,a_1,a_2)(t) \right);
\]
and then obtain $v_s^j$ as the solution of the matrix games, i.e., 
\[
v_s^j= \sup_{d_1 \in \distr(\mov_1(s))} \inf_{d_2 \in \distr(\mov_2(s))} 
\sum_{a_1 \in \mov_1(s),a_2\in \mov_2(s)} d_1(a_1) \cdot d_2(a_2) \cdot M_s^j(a_1,a_2). 
\]
The above algorithm is refered to as the \emph{value-iteration} algorithm.
It is well-known that $v_s = \lim\inf_{T \to\infty} v_s^T =\lim\sup_{T \to\infty} v_s^T$~\cite{MN81},
i.e., the value of the finite-horizon games converge to the value of the game.
We first establish a result that shows that for all $T$ there exist $s$ and $s'$ such that 
$v^*$ is bounded by $v_s^T$ and $v_{s'}^T$. 

\begin{lemma}\label{lemm:finite-horizion}
For all ECMPGs $G$ and for all $T>0$, there exists a pair of states $s',s$, 
such that $v^T_{s'}\leq v^* \leq v^T_{s}$.
\end{lemma}
\smallskip\noindent{\bf Proof overview:} 
The proof is by contradiction, that is, we assume that for all $s$ we have $v^T_s < v^*$ (the other case follows 
from the same game where the players have exchanged roles). 
The idea is that we can consider plays of $G$, defined by an optimal strategy for the objective 
$\LimInfAvg$ for Player~1 in $G$ and a Markov strategy for Player~2 that plays an optimal Markov 
strategy for objective $\Avg_T$ in $G_T$ for $T$ steps and then starts over. 
We then split the plays into sub-plays of length $T$. 
Since for all $s$ we have $v^T_s < v^*$  and because Player~2 plays optimally in the sub-plays,
in every segment of length $T$ the expected mean-payoff is strictly less than $v^*$. 
But then also the expected mean-payoff of the plays is strictly less than $v^*$. 
This contradicts that Player~1 played optimally 
(which ensures that the expected mean-payoff is at least $v^*$).
We now present the formal proof of the lemma.

\begin{proof}
We argue explicitly about $v^*\leq v^T_{s}$ and the other inequality follows by considering the same game, 
but where the players have exchanged roles.
Assume towards contradition that there exists an ECMPG $G$, a time-bound $T>0$, 
and an $\epsilon>0$, such that for all $s$ we have $v^* \geq  v^T_{s}+\epsilon$.

Let $\sigma_2'$ be an optimal Markov strategy for $\Avg_T$ in the finite-horizion game $G_T$, 
and let $\sigma_2$ be the Markov strategy in $G$ defined as follows: 
for plays of length $T'$ with last state $s'$ we have 
$\sigma_2(T',s')=\sigma_2'(T'\mod T,s')$, for all $T' \geq 0$ and states $s'$. 
Let $\sigma_1$ be an $\epsilon/4$-optimal strategy for the objective $\LimInfAvg$ for Player~1.
Then for all $s$ we have 
\[
v^* -\epsilon/4\leq \Exp_{s}^{\stra_1,\stra_2}[\LimInfAvg]=\Exp_{s}^{\stra_1,\stra_2}[\liminf_{T'\rightarrow \infty}\Avg_{T'}]\leq\liminf_{T'\rightarrow \infty}\Exp_{s}^{\stra_1,\stra_2}[\Avg_{T'}] \enspace ,
\]
where we used that $\stra_1$ is $\epsilon/4$-optimal in the first inequality and Fatou's lemma in the second inequality.
Let  $T'$ be such that $\Exp_{s}^{\stra_1,\stra_2}[\Avg_{T'}] > v^* -\epsilon/2$ and $T' \mod T \equiv 0$; 
by the preceding expression there exists $T_0$ such that for all $T_1 \geq T_0$ we have 
$\Exp_{s}^{\stra_1,\stra_2}[\Avg_{T_1}] > v^* -\epsilon/2$ and hence such $T'$ always exists.
Let $\Theta_i$ be the random variable denoting the $i$-th state and action pairs $(s_i,a_1^i, a_2^i)$.
From the definition of $T'$ we get that 
\begin{align*}
v^* -\epsilon/2 <\Exp_{s}^{\stra_1,\stra_2} [\Avg_{T'}]
=\Exp_{s}^{\stra_1,\stra_2}\left[ \frac{1}{T'} \cdot \sum_{i=0}^{T'-1} \cost(\Theta_i)\right]
& = 
\Exp_{s}^{\stra_1,\stra_2}\left[\frac{T}{T'} \cdot \sum_{j=0}^{T'/T-1}\frac{1}{T} \cdot \sum_{i=0}^{T-1} \cost(\Theta_{i+j\cdot T})\right] \\[3ex]
& = \frac{T}{T'} \cdot \sum_{j=0}^{T'/T-1} \Exp_{s}^{\stra_1,\stra_2}\left[\frac{1}{T}  \cdot \sum_{i=0}^{T-1} \cost(\Theta_{i+j\cdot T})\right]
\enspace ,
\end{align*}
where the first equality is expanding the definition of $\Avg_T$;
the second equality is obtained by splitting the sum into sub-sums of length $T$; 
and the third equality is by linearity of expectation.
For any $j$, the number 
\[c_j=  \Exp_{s}^{\stra_1,\stra_2}\left[\frac{1}{T}\sum_{i=0}^{T-1} \cost(\Theta_{i+j\cdot T})\right]\] 
is at most $v^T_{s_{j\cdot T}}$, because $\sigma_2$  in round $i+j\cdot T$, for $0\leq i < T$, played as 
$\sigma_2'(i,s')$ which is optimal for $\Avg_T$ in $G_T$ 
(note that $c_j$ would be precisely $v^T_{s_{j\cdot T}}$ if also $\sigma_1$ in round $i+j\cdot T$, for $0\leq i < T$, 
played as an optimal strategy for $\Avg_T$ in $G_T$, by definition of the value in $G_T$). 
Note that by the assumption (towards contradiction) we have 
$v^T_{s_{j\cdot T}} \leq v^* -\epsilon$, and hence $c_j \leq v^* -\epsilon$.
Therefore, \begin{align*}
v^*-\epsilon/2 &<\frac{T}{T'} \cdot \sum_{j=0}^{T'/T-1}c_j
\leq \frac{T}{T'} \cdot \sum_{j=0}^{T'/T-1}(v^*-\epsilon)= v^*-\epsilon
\enspace ,
\end{align*}
which is a contradiction.
The desired result follows.
\end{proof}

Note that the proof of the above lemma does not require that the game is ergodic.
The lemma is easily extended to general CMPGs by considering $s$ in the proof 
to be the state of the highest value $v_s$. 

\smallskip\noindent{\bf The numbers $\UH$ and $\OH$.} 
Given an ECMPG $G$, strategies $\stra_1$ and $\stra_2$ for the players,
and two states $s$ and $t$, let $H_{s,t}^{\stra_1,\stra_2}$ denote the 
expected hitting time from $s$ to $t$, given the strategies.
Let $H_{\stra_1} = \sup_{\stra_2 \in \bigstra_2} \max_{s,t \in S} H_{s,t}^{\stra_1,\stra_2}$;
and $\UH = \inf_{\stra_1 \in \bigstra_1} H_{\stra_1}$ and 
$\OH = \sup_{\stra_1 \in \bigstra_1} H_{\stra_1}$.
Intuitively, $\UH$ is the minimum expected hitting time between all 
state pairs that Player~1 can ensure against all strategies of Player~2.

\begin{lemma}
For all ECMPGs $G$ we have $\UH \leq \OH \leq n \cdot (\delta_{\min})^{-\rand}$.
\end{lemma}
\begin{proof}
Since $G$ is ergodic for all strategy profiles and for all state pairs $s$ and $t$,
the state $t$ is reached from $s$ with positive probability (by definition of 
ergodicity).
Hence the desired result follows from Lemma~\ref{lemm:prob to reach}. 
\end{proof}

We now present our main result for the bounds required for approximation by 
the value-iteration algorithm.

\begin{theorem}\label{thrm:val-iter}
For all ECMPGs, for all $0< \epsilon<1$, and all $T\geq 4\cdot \UH \cdot c\cdot \log c$, 
for $c= 2\cdot \epsilon^{-1}$, 
we have that $v^* -\epsilon \leq \min_s v^T_s \leq v^* \leq \max_s v^T_s\leq v^* +\epsilon$.
\end{theorem}
\begin{proof}
Let $T\geq 4\cdot \UH \cdot c\cdot \log c$, for $c= 2\cdot \epsilon^{-1}$. 
Also, let $c'=4\cdot \UH\cdot \log c$ and $T'=T-c'$. 
By Lemma~\ref{lemm:finite-horizion} we have  $\min_s v^{T}_s \leq v^* \leq \max_s v^{T}_s$. 
We now argue that $v^*-\epsilon \leq \min_s v^{T}_s$, and then $\max_s v^T_s\leq v^* + \epsilon$ 
follows by considering the game where the players have exchanged roles.

Let $s'$ be some state in $\arg\min_{s'} v^T_{s'}$ and let $s''$ be some state such that $v^* \leq v^{T'}_{s''}$ 
(such a state exists by Lemma~\ref{lemm:finite-horizion}). 
Let $\sigma_1'$ be an optimal Markov strategy for the objective $\Avg_{T'}$ in $G_{T'}$, and 
let $\sigma_1^*$ be a strategy that ensures that the hitting time from $s'$ to $s''$ is at most $2\cdot \UH$, i.e.,
$H_{\stra_1^*}\leq 2\cdot \UH$ (such a strategy exists by definition of $\UH$).
Let  $\sigma_1$ be the strategy for Player~1 that plays as $\sigma_1^*$ until $s''$ is reached, 
and then switches to $\sigma_1'$.
Formally, until $s''$ is reached it plays as $\sigma_1^*$, and 
in the $i$-th round after reaching $s''$ the first time, 
if the play is in state $s$, the strategy $\sigma_1$ uses the distribution $\sigma_1'(i,s)$, 
for each $0\leq i<T'$, and after $T'$ steps  since the first visit to $s''$ the strategy $\sigma_1$ plays arbitrarily. 
Let $\sigma_2$ be an arbitrary strategy for Player~2. 

At any point before reaching $s''$, we have that the probability that we do not reach $s''$ within the next  $2\cdot H_{\stra_1^*}\leq 4\cdot \UH$ steps is at most $\frac{1}{2}$ by Markov's inequality. 
Therefore, the probability that we do not reach $s''$ within the first $c'=4\cdot \UH\cdot \log c$ steps is at most $c^{-1}$. 
We now consider two cases, either (1)~we do not reach $s''$ within $c'$ steps; or (2)~we do reach $s''$ within $c'$ steps.  
In case~(1) we get a mean-payoff of at least 0 (since all payoffs are at least 0). 
In case~(2) we split plays up in three parts: (i)~before reaching $s''$; (ii)~the first $T'$ steps after reaching $s''$; (iii)~the rest. 
The expected mean-payoff of the first and the last part is at least 0 and the expected mean-payoff of part (ii) is at least $v^{T'}_{s''}\geq v^*$, 
by definition of $s''$ and $\sigma_1$. 
We now conclude that the expected mean-payoff is at least 
\begin{align*}
(1-c^{-1}) \cdot \frac{T'\cdot v^{T'}_{s''}}{T} &\geq (1-c^{-1}) \cdot \frac{(T-c')\cdot v^*}{T} \\
&= v^*-c^{-1}\cdot v^*-(1-c^{-1})\cdot\frac{c'}{T}\cdot v^*\\
&\geq  v^* -c^{-1}- c^{-1} \qquad \text{(since $v^* \leq 1$ and $T=c' \cdot c$)}; \\
&=  v^*-\epsilon \enspace .
\end{align*}
Therefore against any strategy $\sigma_2$, the strategy $\sigma_1$ ensures at least $v^*-\epsilon$ 
for $\Avg_T$ in $G_T$. 
This is then also true for all optimal strategies for Player~1 for $\Avg_T$ in $G_T$ and thus the result follows.
\end{proof}

\begin{remark}\label{remark:val-iter}
Theorem~\ref{thrm:val-iter} presents the bound for value-iteration when the rewards are in 
the interval $[0,1]$. 
If the rewards are in the interval $[0,W]$, for some positive integer $W$, then for 
$\epsilon$-approximation we first divide all rewards by $W$, and then apply the results
of Theorem~\ref{thrm:val-iter} in the resulting game for $\epsilon/W$-approximation. 
We have shown that in the worst case $\UH$ is at most $n \cdot (\delta_{\min})^{-\rand}$.
If $\UH,W,\epsilon^{-1}$ are bounded by a polynomial, then the 
value-iteration algorithm requires polynomial-time to approximate;
and hence if $\UH$ and $W$ are bounded by polynomial, then the value-iteration 
algorithm is a $\FPTAS$. 
In particular, if either (i)~$\rand$ is constant and $(\delta_{\min})^{-1}$ is bounded by a polynomial, 
or (ii)~$(\delta_{\min})^{-1}$ is bounded by a constant and $\rand$ is logarithmic in $n$,
then $\UH$ is polynomial; and if $W$ is polynomial as well, then the value-iteration algorithm 
is a $\FPTAS$.
There could also be other cases where $\UH$ is polynomial, and then the value-iteration 
is a pseudo-polynomial time algorithm for constant-factor approximation.
\end{remark}

\section{Exact Value Problem for Almost-sure Ergodic Games}\label{sec:precise ergodic}
We present two results related to the exact value problem:
(1)~First we show that for almost-sure ergodic CMPGs the exact value 
can be expressed in the existential theory of the reals;
and (2)~we establish that the value problem for sure ergodic CMPGs 
is square-root sum hard.

\subsection{Value problem in existential theory of the reals}
We show how to express the value problem for almost-sure ergodic CMPGs 
in the existential theory of the reals (with addition and multiplication) 
in three steps (for details about the existential theory of the reals 
see~\cite{Canny88,BPR}).
In the general case of CMPGs the current known solution for the value problem 
in the theory of reals uses three quantifier alternations~\cite{CMH08}, and 
in the theory of reals one of the computationally expensive step is the 
quantifier alternation elimination.

\smallskip\noindent{\em Step~1: Ergodic decomposition computation.}
First we compute the ergodic decomposition of an almost-sure ergodic CMPG
in polynomial time,
and let $C_1, C_2,\ldots, C_\ell$, be the $\ell$ ergodic components.
The polynomial time algorithm is as follows: construct a graph with state 
space $S$, and put an edge $(s,t)$ iff $t$ is reachable from $s$ in the 
CMPG.
The bottom scc's of the graph are the ergodic components, where a bottom scc
is an scc with no out-going edges leaving the scc.

\smallskip\noindent{\em Step~2: Existential theory of the reals sentence
for an ergodic component.}
For an ergodic CMPG $G$, Hoffman-Karp~\cite{HK} shows that the value is the unique fixpoint 
of the strategy-iteration algorithm. 
The algorithm iteratively takes a strategy $\sigma_1$ for Player~1, 
computes the optimal best response strategy $\sigma_2$ for Player~2, and 
computes the potentials of each state $v_s^{\sigma_1}$ and the value $g^{\sigma_1}$ 
guaranteed by $\sigma_1$. A strategy for Player~1 that ensures a higher value than 
$g^{\sigma_1}$ is then, for every state $u$, to use an optimal distribution in the matrix game defined 
by $M[a_1,a_2]=\cost(u,a_1,a_2)+ \sum_{s \in S} \trans(u,a_1,a_2)(s) \cdot v^{\sigma_1}_s$.
We quantify over stationary strategies in the existential theory of the reals, and use
the following notation: for a set $\set{x_1,x_2,\ldots,x_k}$ of variables we write
$\ProbDist(x_1,x_2,\ldots,x_k)$ to denote the constraints (i)~$x_i \geq 0$ for $1 \leq i \leq k$, 
and (ii)~$\sum_{i=1}^k x_i=1$; which specifies that the set of variables forms 
a probability distribution. 
We can formulate the fixpoint of the Hoffman-Karp algorithm (and thus the value $g$) 
using existential first order theory as follows. 
Fix a specific state $s^*$, and then consider the following sentence 
where we quantify existentially over the variables 
$g,(x_{s,i})_{s\in S,i\in \Gamma_1(s)},(y_{s,j})_{s\in S,j\in \Gamma_2(s)},(v_s)_{s\in S}$, 
have the following constraints:
\begin{align}
& \Phi(g,(x_{s,i})_{s\in S,i\in \Gamma_1(s)},(y_{s,j})_{s\in S,j\in \Gamma_2(s)},(v_s)_{s\in S}) = \\
& \qquad 
\bigwedge_{s\in S} \bigwedge_{j\in \Gamma_2(s)} (g +v_s \leq \sum_{i\in \Gamma_1(s)} (x_{s,i} \cdot (\cost(s,i,j) + \sum_{t\in S}(\trans(s,i,j)(t) \cdot v_t))) \ \ \wedge \label{exp:value is less than} \\[2ex]
& \qquad \bigwedge_{s\in S} \bigwedge_{i\in \Gamma_1(s)} (g +v_s \geq \sum_{j\in \Gamma_2(s)} (y_{s,j} \cdot (\cost(s,i,j) + \sum_{t\in S}(\trans(s,i,j)(t) \cdot v_t))) \ \ \wedge \label{exp:value is greater than} \\[2ex]
& \qquad \bigwedge_{s\in S} \ProbDist(x_{s,1},x_{s,2} \ldots,x_{s,|\Gamma_1(s)|}) \land \bigwedge_{s\in S} \ProbDist(y_{s,1},y_{s,2} \ldots,y_{s,|\Gamma_2(s)|}) \ \ \wedge \\[2ex]
& \qquad (v_{s^*}=0) \enspace .
\end{align} 
Notice that $(\ref{exp:value is less than})$ and the fact that the variables $x_{s,i}$ gives a probability distribution, ensures that 
$x_{s,i}$ gives an optimal strategy in the matrix game of potentials, similar for $(\ref{exp:value is greater than})$ and $y_{s,j}$. 
Also, $(\ref{exp:value is less than})$ and $(\ref{exp:value is greater than})$ implies that 
\[\bigwedge_{s\in S} (g+v_s= \sum_{j\in \Gamma_2(s)}\sum_{i\in \Gamma_1(s)}(y_{s,j}\cdot x_{s,i}\cdot (\cost(s,i,j)+ \sum_{t\in S}(\trans(s,i,j)(t)\cdot v_t)))\enspace ,\] 
which together with $(\ref{exp:value is less than})$ ensures that 
\[\forall s: g +v_s = \max_{j \in \Gamma_2(s)} \sum_{i\in \Gamma_1(s)} (x_{s,i} \cdot (\cost(s,i,j) + \sum_{t\in S}(\trans(s,i,j)(t) \cdot v_t)))\enspace .\] 
The preceding equality together with $(v_{s^*}=0)$ ensures that $(v_s)_{s\in S}$ is the potential associated with the stationary strategy $x$, 
and hence, $g$ is the value of the game. 
The sentence $\Phi$ in the existential theory of the reals for the value is
\[
\exists g,(x_{s,i})_{s\in S,i\in \Gamma_1(s)},(y_{s,j})_{s\in S,j\in \Gamma_2(s)},(v_s)_{s\in S}:
\Phi(g,(x_{s,i})_{s\in S,i\in \Gamma_1(s)},(y_{s,j})_{s\in S,j\in \Gamma_2(s)},(v_s)_{s\in S}); \] 
and $g$ denotes the value of the component.

\smallskip\noindent{\em Step~3: Existential theory of the reals sentence
for an almost-sure ergodic CMPG.}
Given a real number $\lambda$ and an almost-sure ergodic CMPG $G$, 
we now give an existential theory of the reals sentence, which can be satisfied iff $G$ has value at most $\lambda$.
Let $C_1, C_2,\ldots,C_\ell$ be the ergodic components, and let $C=\bigcup_{i=1}^\ell C_i$.
We denote by $\Phi_{C_i}$ the existential theory of the reals sentence for the value 
in component $C_i$ (as described in Step~2) and the variable $g_i$ is the value.
The existential theory sentence for other states is given using the formula 
for reachability games.
We quantify existential over the variables $((z_s)_{s \in S},(x_{i,s})_{s\in (S \setminus C),i\in \Gamma_1(s)},(y_{s,j})_{s\in (S\setminus C),j\in \Gamma_2(s)})$
and have the following constraints:
\[
\begin{array}{lll}
&\bigwedge_{1\leq i\leq \ell} &\Phi_{C_i} \ \ \wedge \\[2ex]
\bigwedge_{s\in (S\setminus C)}&\bigwedge_{j\in \Gamma_2(s)}&(z_s\leq \sum_{i\in \Gamma_1(s)} \sum_{t\in S} x_{s,i}\cdot\trans(s,i,j)(t)\cdot z_t) \ \ \wedge  \\[2ex]
\bigwedge_{s\in (S\setminus C)}&\bigwedge_{i\in \Gamma_1(s)}&(z_s\geq \sum_{j\in \Gamma_2(s)} \sum_{t\in S} y_{s,j}\cdot\trans(s,i,j)(t)\cdot z_t) \ \ \wedge  \\[2ex]
\bigwedge_{1 \leq i \leq \ell} &\bigwedge_{s \in C_i} & (z_s= g_i)  \ \ \wedge \\[2ex]
&\bigwedge_{s\in (S \setminus C)} &\ProbDist(x_{s,1},x_{s,2} \ldots,x_{s,|\Gamma_1(s)|}) \land \bigwedge_{s\in S} \ProbDist(y_{s,1},y_{s,2} \ldots,y_{s,|\Gamma_2(s)|}) \ \ \wedge \\[2ex]
&&(z_s\leq \lambda)  \enspace .
\end{array} 
\]
The idea is as follows: 
First note that the constraint $z_s=g_i$, for $s \in C$, ensures that for all states in the ergodic
component the variable $z_s$ denotes the value of $s$ (by the correctness of the formula $\Phi_{C_i}$ for 
an ergodic component $C_i$).
If the value of state $s \in (S\setminus C)$ in $G$ is $z_s$, for all $s$,  
then \[\bigwedge_{s\in (S\setminus C)}\bigwedge_{j\in \Gamma_2(s)}(z_s\leq \sum_{i\in \Gamma_1(s)} \sum_{t\in S} x_{s,i}\cdot\trans(s,i,j)(t)\cdot z_t)\] 
ensures that $x$ is an optimal strategy in the game. 
Also, similar to the ergodic part, 
\[\bigwedge_{s\in (S\setminus C)}\bigwedge_{j\in \Gamma_2(s)}(z_s\leq \sum_{i\in \Gamma_1(s)} \sum_{t\in S} x_{s,i}\cdot\trans(s,i,j)(t)\cdot z_t)\ ; \  \bigwedge_{s\in (S\setminus C)}\bigwedge_{i\in \Gamma_1(s)}(z_s\geq \sum_{j\in \Gamma_2(s)} \sum_{t\in S} y_{s,j}\cdot\trans(s,i,j)(t)\cdot z_t)\] 
implies that for all $s$:\[(z_s=\max_{j\in \Gamma_2(s)}\sum_{i\in \Gamma_1(s)}\sum_{t\in S} x_{s,i}\cdot \trans(s,i,j)(t)\cdot z_t) \enspace .\]
Therefore, the vector $\bar{z}$, such that $\bar{z}_s=z_s$ is a fixpoint for the value-iteration algorithm for 
reachability objectives. 
Hence, the fact that $z_s\leq \lambda$, implies that the least fixpoint $\tilde{z}$ of the value-iteration algorithm (which is the value of the game) 
is such that $\tilde{z}_s\leq \lambda$. 
Thus, we get the following theorem.

\begin{theorem}
The value problem for almost-sure ergodic CMPGs can be expressed in the existential
theory of the reals.
\end{theorem}

\subsection{Square-root sum hardness}
In this section we show that the value problem for sure ergodic CMPGs is at least
as hard as the square-root sum problem by generalizing the example we presented 
in Figure~\ref{fig:ergodic sqrt intro}.

\smallskip\noindent{\bf Square-root sum problem.}
The \emph{square-root sum problem} is the following decision problem: 
Given a positive integer $v$ and a set of positive integers $\{n_1,\dots,n_\ell\}$, is $\sum_{i=1}^\ell \sqrt{n_i}\geq v$? 
The problem is known to be in the the fourth level of the counting hierarchy~\cite{ABKPM09}, 
but it is a long-standing open problem if it is in $\NP$.

\smallskip\noindent{\bf Reduction to sure ergodic CMPGs.} 
The reduction is similar to~\cite{EY,EY10}. 
First we define a family of ergodic 
CMPGs $\{G_b\mid b\in \N\}$, such that $G_b$ has value $\sqrt{b}$.  
Given an instance of the square-root sum problem, $(v,\{n_1,\dots,n_\ell\})$, 
we use our family to get an ergodic CMPG $G_{n_i}$ for each number $n_i$. 
We use one more state $s^*$, with one action for each player. 
The successor of $s^*$ is $G_{n_i}$ with probability $\frac{1}{\ell}$ for every $i$. 
This ensures that the value of $s^*$ is $\frac{\sum_i \sqrt{n_i}}{\ell}$. 
Thus, the value of $s^*$ is at least $\frac{v}{\ell}$ iff $\sum_i \sqrt{n_i}\geq v$. 
Notice that we reach an ergodic component in precisely one step from $s^*$, 
and thus the game is sure ergodic.

\smallskip\noindent{\em The numbers $k_b$ and $d_b$.}
First we  define $G_b$, for $b\not \in \{1,2,4\}$. We define $G_b$ for $b\in \{1,2,4\}$ afterwards. 
To define $G_b$ for $b\not\in \{1,2,4\}$, we use two numbers $k_b$ and $d_b$, such that $k_b>d_b>0$, defined as follows: 
Let $k_b$ be the smallest positive integer such that $k_b^2>b$. 
Let $d_b=2\cdot k_b-\frac{2\cdot b}{k_b}$, implying that $b=k_b^2-\frac{d_b\cdot k_b}{2}$. 
This gives us directly that $d_b>0$  (and hence also $\frac{d_b\cdot k_b}{2}\in \N$).  
We show that $k_b>d_b$. First, for $b=3$, we see that $k_3$ is $2$ and $3=2^2-\frac{1\cdot 2}{2}$ and thus  $d_3=1$, 
implying that $k_3>d_3$. For $9>b\geq 5$, we see that $k_b=3$ and $d_b\in [\frac{2}{3};\frac{8}{3}]$ and again have that $k_b>d_b$. 
For $b\geq 9$, we show the statement using contradiction. Assume therefore that $d_b\geq k_b$. 
We then get that $b=k_b^2-\frac{d_b\cdot k_b}{2}\Rightarrow b\leq \frac{k_b^2}{2}$. 
By definition of $k_b$ we know that $b\geq (k_b-1)^2=k_b^2+1-2\cdot k_b\geq k_b^2+1-\frac{k_b}{2}\cdot k_b>\frac{k_b^2}{2}$. 
That is a contradiction. The second to last inequality is because for $b\geq 9$, we have that $k_b\geq 4$. Thus, $k_b>d_b$ for $b\not\in \{1,2,4\}$.

\smallskip\noindent{\em Construction of $G_b$.}
For a positive integer $b\not \in \{1,2,4\}$, we define $G_b$ as follows. There are two states in $G_b$, $u$ and $w$. The state $w$ has a single action for Player~1 and a single action for Player~2, $a_w$ and $b_w$ respectively, and the successor of $w$ is always $u$. Also $\cost(w,a_w,b_w)=k_b$. The state $u$ has two actions for each of the two players. Player~1 has actions $a_u^1$ and $a_u^2$. Player~2 has actions $b_u^1$ and $b_u^2$. For any pair of actions $a_u^i$ and $b_u^j$ we have that the successor, $\delta(u,a_u^i,b_u^j)$ is $w$, except for $a_u^1$ and $b_u^1$ for which the successor is $u$ with probability $\frac{d_b}{k_b}$ and $w$ with probability $1-\frac{d_b}{k_b}$. Note that $\frac{d_b}{k_b}$ is a number in $(0,1)$, since $k_b>d_b>0$. The rewards $\cost(u,a_u^1,b_u^2)=\cost(u,a_u^2,b_u^1)$ are $k_b-d_b$. The rewards $\cost(u,a_u^1,b_u^1)=\cost(u,a_u^2,b_u^2)$ are $k_b$. The game is ergodic, since $\frac{d_b}{k_b}<1$, and thus there is a positive probability to change to the other state in every step, no matter the choice of the players. There is an illustration of $G_b$ in Figure~\ref{fig:ergodic sqrt}.

\begin{figure}
\begin{center}
\begin{tikzpicture}[node distance=3cm]
\matrix (v1) [label=right:$w$,minimum height=1.5em,minimum width=1.5em,matrix of math nodes,nodes in empty cells, left delimiter={.},right delimiter={.}]
{
\\
};
\draw[black] (v1-1-1.north west) -- (v1-1-1.north east);
\draw[black] (v1-1-1.south west) -- (v1-1-1.south east);
\draw[black] (v1-1-1.north west) -- (v1-1-1.south west);
\draw[black] (v1-1-1.north east) -- (v1-1-1.south east);
\matrix (s) [label=left:$u$,left of=v1,minimum height=1.5em,minimum width=1.5em,matrix of math nodes,nodes in empty cells, left delimiter={.},right delimiter={.}]
{
&\\
&\\
};
\draw[black] (s-1-1.north west) -- (s-1-2.north east);
\draw[black] (s-1-1.south west) -- (s-1-2.south east);
\draw[black] (s-2-1.south west) -- (s-2-2.south east);
\draw[black] (s-1-1.north west) -- (s-2-1.south west);
\draw[black] (s-1-2.north west) -- (s-2-2.south west);
\draw[black] (s-1-2.north east) -- (s-2-2.south east);

\draw[->](s-1-1.center) .. controls +(45:1) and +(90:1).. node[midway,above] (x) {$\frac{d_b}{k_b}$} (s-1-1.north);
\draw[->,out=45,in=130](s-1-1.center) 
 to node[midway,above] (x) {$1-\frac{d_b}{k_b}$} (v1);
\draw[->,out=65,in=155,dashed](s-1-2.center) to (v1);
\draw[->,out=-65,in=-155](s-2-2.center) to (v1);
\draw[->,out=-45,in=-130,dashed](s-2-1.center) to (v1);
\draw[->](v1.center) to (s);
\end{tikzpicture}
\caption{The game $G_b$, such that $b=k_b^2-\frac{k_b\cdot d_b}{2}$.  Dashed edges has reward $k_b-d_b$ and non-dashed edges has reward $k_b$. 
Actions are annotated with probabilities if the probability is not 1.\label{fig:ergodic sqrt}}
\end{center}
\end{figure}
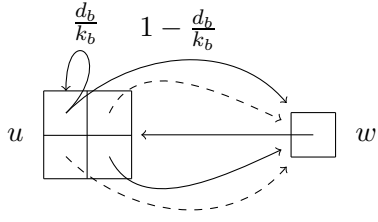

\begin{remark}
For $b\not \in \{1,2,4\}$, the numbers $k_b$ and $d_b$ have short binary descriptions. The number $k_b>0$ cannot be larger than $\sqrt{2\cdot b}$, because otherwise $k_b^2-\frac{d_b\cdot k_b}{2}\geq \frac{k_b^2}{2}>b$. It must also be a  positive integer and thus has a binary representation of length at most $\frac{1+\log b }{2}$. Also $k_b>d_b>0$ and $\frac{d_b\cdot k_b}{2}$ is a positive integer and thus, $d_b$ has a binary representation of length at most $\frac{1+\log b }{2}+\frac{1+\log b }{2}=1+\log b$. 
\end{remark}

\smallskip\noindent{\em $G_b$ for $b\in \set{1,2,4}$.}
One can, using the preceding, define $G_b$ for all  positive integers $b$ which is not in $\{1,2,4\}$. It is also easy to construct games, which has value $\sqrt{1}$ and $\sqrt{4}$, since they are integers. Let $G_1$ be an arbitrary ergodic CMPG of value $1$ and $G_4$ be an arbitrary ergodic CMPG of value $2$. One can also construct a ergodic CMPG, which has value $\sqrt{2}$,  similar to our construction of $G_b$ for $b\not \in \{1,2,4\}$, using fractional\footnote{We do not use fractional $k_b$ in general only because it becomes harder to argue that the games has  a polynomial length binary representation.} $k_2$ and $d_2$. We see that $k_2=\frac{3}{2}$ and $d_2=\frac{1}{3}$ gives us that $2=k_2^2-\frac{d_2\cdot k_2}{2}$, while ensuring that $k_2>d_2>0$. Let $G_2$ be the game defined analogous to $G_b$ for $b\not \in\{1,2,4\}$ using $k_2=\frac{3}{2}$ and $d_2=\frac{1}{3}$.

\smallskip\noindent{\em The value in $G_b$ is $\sqrt{b}$.}
We now argue that for a fixed $b\not \in \{1,4\}$, the game $G_b$ has value $\sqrt{b}$ (by definition, the CMPGs $G_1$ and $G_4$ had value $1$ and $2$ resp.). 
We use that $b=k_b^2-\frac{d_b\cdot k_b}{2}$ and that $k_b>d_b>0$. Let $\sigma_1$ be some arbitrary stationary optimal strategy for Player~1.  
Let $p$ be the probability that $\sigma_1$ plays $a_u ^1$. Let $a$ be the optimal potential of state $u$, then the potential of $w$ is 0. 
Let $v$ be the value of $G_b$. Then as shown by Hoffman-Karp~\cite{HK} the strategy $\sigma_1$ must satisfy the equation system
\begin{align*}
a &=p\cdot (k_b-d_b)+(1-p)\cdot k_b-v\\
a &=(1-p)\cdot (k_b-d_b)+p\cdot k_b+\frac{d_b\cdot p \cdot a}{k_b}-v\\
0 &=a+k_b-v\\
\end{align*}
From the third equation we obtain $a=v-k_b$, and substituting in the first
equation we obtain that 
\[
2\cdot k_b =p\cdot d_b+2\cdot v \quad \Rightarrow  p=\frac{2\cdot k_b-2\cdot v}{d_b}
\]
Substituting $a$ and $p$ from above into the second equation we obtain
\begin{align*}
0 &=2\cdot k_b-2\cdot v-d_b+2\cdot k_b-2\cdot v+\frac{d_b\cdot (2\cdot k_b-2\cdot v) \cdot (v-k_b)}{k_b\cdot d_b}\\
\Rightarrow 
0 & =2\cdot k_b-d_b-\frac{2\cdot v^2}{k_b}\\
\Rightarrow 
0 &=\frac{k_b^2}{2}-\frac{d_b\cdot k_b}{4}-\frac{v^2}{2} \qquad \text{(Multiply by $k_b$ and divide by $4$).}
\end{align*}
Solving the above second degree equation for $v$ we obtain that 
\[
v =\frac{-0\pm\sqrt{-4\cdot (\frac{k_b^2}{2}-\frac{d_b\cdot k_b}{4})\cdot\frac{-1}{2}}}{2\cdot\frac{-1}{2}}
\qquad \Rightarrow
v = \pm\sqrt{b}
\]
Since we know that the value is positive (since all rewards are positive, because $k_b>d_b>0$), 
we see that $v=\sqrt{b}$.
Thus the desired property is established.

\begin{theorem}
The value problem for sure ergodic CMPGs is square-root sum hard.
\end{theorem}

\section{Conclusion}
In this work we established the strategy complexity and the approximation 
complexity for ergodic, sure ergodic, and almost-sure ergodic mean-payoff games.
Our results also show that the approximation problem for turn-based stochastic
ergodic mean-payoff games is at least as hard as the value problem for SSGs.
In contrast, for concurrent deterministic almost-sure ergodic games, the 
value problem can be solved in polynomial time.
In concurrent deterministic games, in every ergodic component all states have
an unique successor, and hence an optimal strategy and the value can be 
computed in polynomial time. 
In any given concurrent deterministic almost-sure ergodic game, once the values of the ergodic components have been computed, the value iteration algorithm 
computes the values for the remaining states in $n$ 
iterations.
Moreover, we established that the value problem for sure ergodic games is 
square-root sum hard.
Note that for sure ergodic games with reachability objectives, the values can 
be computed in polynomial time by value iteration for $n$ iterations.
This shows informally that the hardness of sure ergodic games is due to 
mean-payoff objectives.
Since we have shown that values of ergodic games can be irrational, we 
conjecture that the value problem for ergodic games itself is sqaure-root 
sum hard, but an explicit reduction is likely to be cumbersome.

\end{document}